\documentclass[sts]{imsart}

\usepackage{amsmath}
\usepackage{amssymb}
\usepackage{latexsym}
\usepackage{mathrsfs}
\usepackage{natbib}
\usepackage{hyperref}
\usepackage{amsthm}
\usepackage{graphicx,psfrag,epsf}
\usepackage{enumerate}
\usepackage{caption}
\usepackage{appendix}

\newcommand{\de}{\mathrm{d}}

\newcommand{\x}{{\mathbf x}}
\newcommand{\R}{{\mathbb R}}

\newcommand{\B}{{\mathcal B}}

\newcommand{\1}{{\mathbf 1}}

\def\P{\mathbb P}
\def\E{\mathbb E}

\newtheorem{thm}{Theorem}[section]

\newtheorem{definition}{Definition}[section]

\newtheorem{lemma}{Lemma}[section]
\newtheorem{remark}{Remark}[section]
\newtheorem{cor}{Corollary}[section]

\graphicspath{{./Figures/}}

\begin{document}
\begin{frontmatter}

\title{Inhomogeneous higher-order summary statistics for linear network point processes %\protect\thanksref{T1}
}
\runtitle{Higher-order summary statistics}
%\thankstext{T1}{Footnote to the title with the `thankstext' command.}

\begin{aug}
\author{\fnms{Ottmar} \snm{Cronie}\ead[label=e1]{}},
\author{\fnms{Mehdi} \snm{Moradi}\thanksref{t3}\ead[label=e2]{mehdi.moradi@unavarra.es}}
\and
\author{\fnms{Jorge} \snm{Mateu}
\ead[label=e3]{}\ead[label=u1,url]{}}

\thankstext{t3}{Corresponding author  \printead{e2}}

\runauthor{Cronie et al.}

\affiliation{Ume{\aa} University, Public University of Navarra and University Jaume I}

\address{Department of Mathematics and Mathematical Statistics, Ume{\aa} University, Sweden; Department of Statistics and Operation Research, Public University of Navarra, Pamplona, Spain;}

\address{and Department of Mathematics, University Jaume I, Castell\'{o}n, Spain.}
\end{aug}

\begin{abstract}
We introduce the notion of intensity reweighted moment pseudostationary point processes on linear networks. Based on arbitrary general regular linear network distances, we propose geometrically corrected versions of different higher-order summary statistics, 
including 
%such as 
the inhomogeneous empty space function, the inhomogeneous nearest neighbour distance distribution function and the inhomogeneous $J$-function. We also discuss their non-parametric estimators. Through a simulation study, considering models with different types of spatial interaction, we study the performance of our proposed summary statistics. Finally, we make use of our methodology to analyse two datasets: motor vehicle traffic accidents and spider data.
\end{abstract}

\end{frontmatter}

\section{Introduction}

Nowadays point patterns are sampled on a variety of different spatial domains \citep{BRT15}. In particular, point patterns on linear networks and their associated statistical analysis have gained a considerable amount of interest; Figure \ref{fig:data} illustrates two such datasets, which will be analysed in this paper. 
Non-parametric analyses of linear network point processes usually tend to have two main ingredients: intensity estimation (first order) and estimation of summary statistics, which indicate whether the underlying point process tends to have a clustering/aggregating or inhibiting/regular behaviour. In essence, such a summary statistic reflects different characteristics of the distribution of points around a point of the point process and/or an arbitrary location within the study region. 
Such analyses differ depending on whether one assumes that the underlying point process has a non-constant intensity function, 
which is referred to as 
%i.e.~ under 
inhomogeneity. 
Thusfar, attention has 
mainly been paid to non-parametric estimators for 
% intensity functions and 
second-order summary statistics, such as $K$-functions and pair correlation functions. 
A wide review of different non-parametric estimators for first- and second-order summary statistics of point patterns on linear networks can also be found in \cite{moradi2018spatial}. 
%{\color{blue} the previous paragraph is not well-connected to the next one}
\begin{figure}[!ht]
    \centering
    \includegraphics[scale=.3]{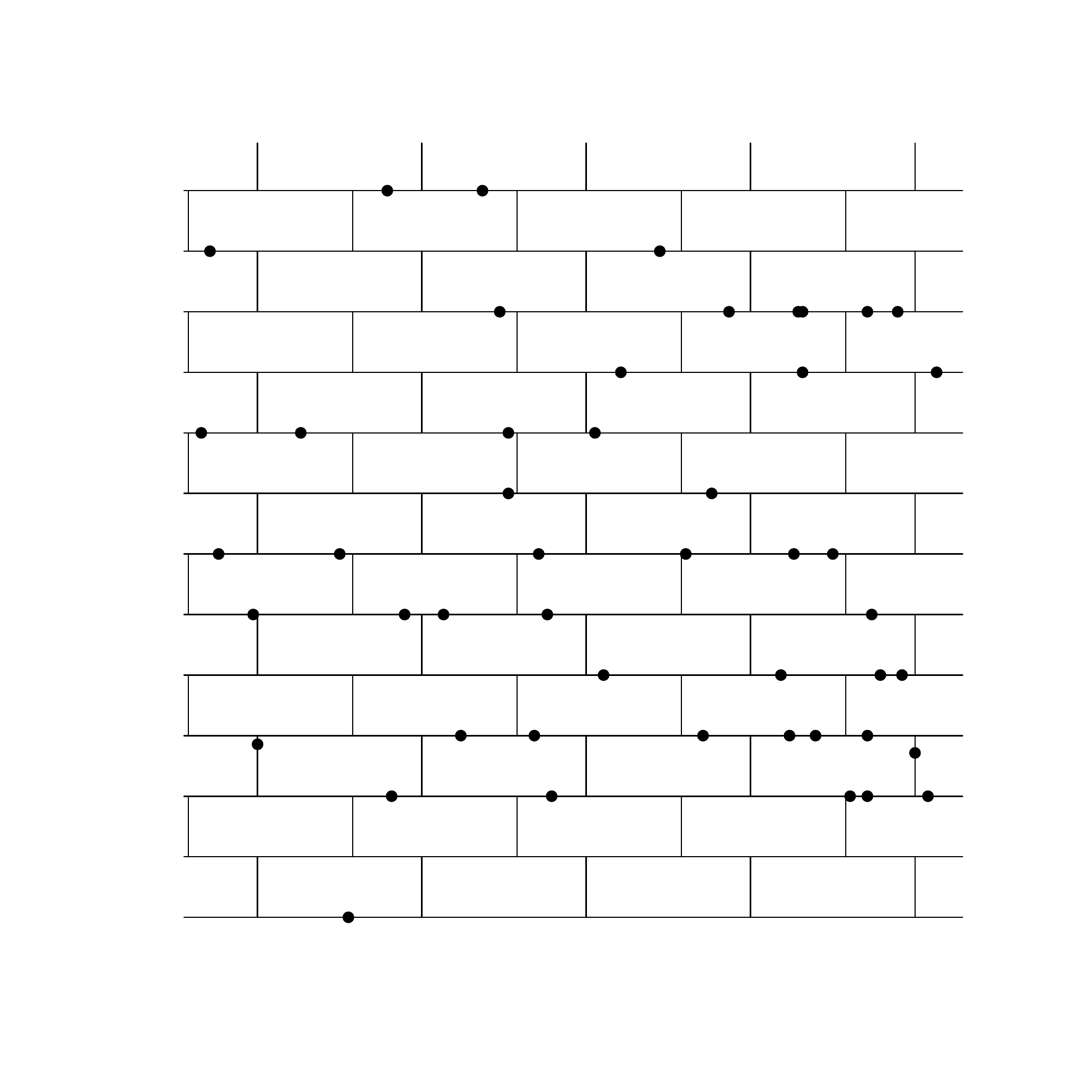}
    \includegraphics[scale=.33]{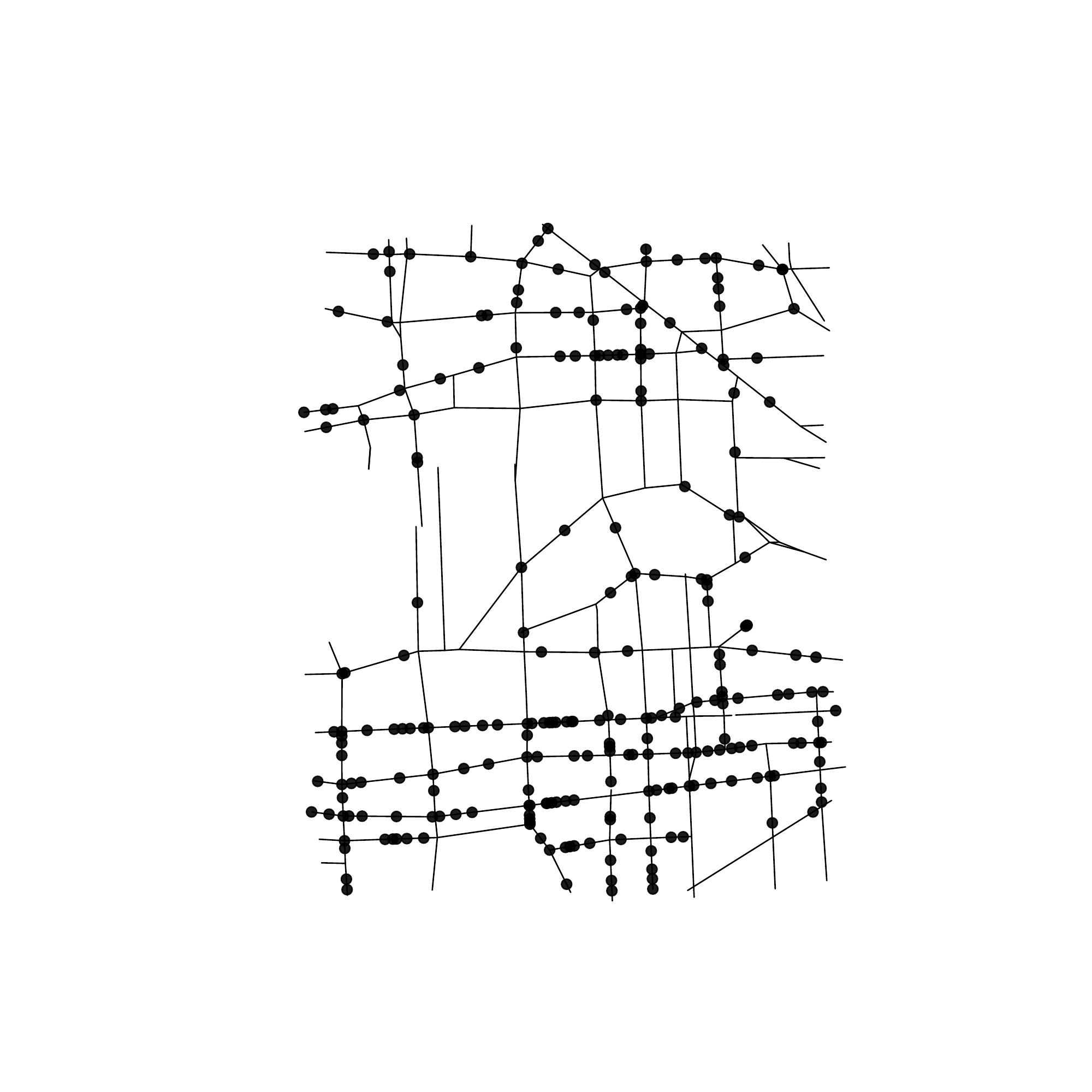}
    \caption{{\em Left}: Spider webs on a brick wall. {\em Right}: Motor vehicle traffic accidents in an area of Houston, US, during April, $1999$.
    }
    \label{fig:data}
\end{figure}
%typical data examples include locations of accidents or crimes on a road network.
% The analysis of point patterns often starts with studying what type of interaction the underlying point process possesses, and one typically looks for indications of 
% %possible 
% clustering/aggregation or inhibition/regularity between the points. An appropriate way to find indications of such behaviour is to study different characteristics of the distribution of points around a data point and/or an arbitrary location within the study region. 
% Such analyses differ depending on whether one assumes that the underlying point process has a non-constant intensity function, i.e.~ under inhomogeneity. 
% Studying point processes on linear network has recently received a lot of attention, but the attention was 
%Regarding statistics for linear network point processes, 
% Thusfar, attention has 
% mostly been paid to non-parametric estimators for intensity functions and second-order summary statistics, such as the $K$-function and the pair correlation function. 
Initially a few poorly performing kernel-based intensity estimators were proposed \citep{B05,B08,XZY08}. Later, other non-parametric kernel-based intensity estimators were defined \citep{OSS09,OS12,mcswiggan2017,MFJ18} which, 
%so that 
although 
%they were 
being statistically well-defined, 
tended to  
%but they might 
be computationally expensive on large networks. Moreover, \cite{rakshit2019fast} proposed a fast kernel intensity estimator based on a  two-dimensional convolution which can be computed rapidly even on large networks. 
With the aim of finding middle-ground between global and local smoothing, as well as an alternative to kernel estimation,  \cite{Moradi2019} introduced their so-called resample-smoothing technique which they applied to Voronoi intensity estimators on arbitrary spaces.
%, regardless of state space, 
They showed that their estimation approach mostly performs 
better than kernel estimators, in terms of bias and standard error. 

% As an alternative to kernel estimation, \cite{Moradi2019} introduced their so-called resample-smoothed Voronoi intensity estimation approach which is defined for point processes on arbitrary spaces,
% %, regardless of state space, 
% and they showed that their estimation approach mostly performs 
% better than kernel estimators, in terms of bias and standard error. 
% statistically, based on bias and variance, behave better than kernel estimators. 

Regarding second-order summary statistics and their estimation, \cite{OY01} considered an analogue of Ripley's $K$-function for homogeneous linear network point processes, which was obtained by using the shortest-path distance instead of the Euclidean distance when measuring distance between points. 
%With respect to second-order summary statistics and in order to propose their counterparts for point patterns on linear networks, \cite{OY01} made use of shortest-path distance instead of Euclidean distance in calculating $K$-function, 
However, this modification did not provide a well-defined $K$-function for 
%point patterns on 
linear network point processes since its behaviour depends on the topography of the network in question. As a remedy, \cite{ABN12} introduced geometrically corrected second-order summary statistics 
%in 
which 
%they 
did not depend on the explicit geometry of the linear network under consideration and has a fixed known behaviour for Poisson processes. 
% topography of the network and have known values for Poisson point processes. 
Hence, the geometrically corrected $K$-function and pair correlation function can be used e.g.~for model selection, hypothesis testing and residual analyses. These summary statistics were later extended to the case of multitype and spatio-temporal point patterns by \citet{baddeley14} and \citet{moradispacetime}. 
Surrounding theses papers, there appeared a discussion on explicit choices of distances to be used for point processes on linear networks -- e.g., is the shortest-path distance always the canonical choice of metric? 
Taking this into account, \cite{rakshit2017second} redefined the $K$- and pair correlation functions of \cite{ABN12} under very general assumptions on the distance/metric used. More specifically, they considered what they referred to as the family of {\em regular distance metrics}.
%Considering different types of distance metrics in the $K$- and pair correlation functions were also proposed by \cite{rakshit2017second}. 
% A wide review of different non-parametric estimators for first- and second-order summary statistics of point patterns on linear networks can also be found in \cite{moradi2018spatial}. 

Although second-order summary statistics are invaluable tools to analyse interaction among points of a point process, the point process may show structure beyond pairwise interactions. For this reason, in the case of point processes in $\R^d$, $d\geq1$, it is common to also study the (inhomogeneous) empty space function and the (inhomogeneous) nearest neighbour distance distribution function  \citep{MW04,IPSS08,van11,CSKWM13,BRT15}. 
%These summary statistics actually deal with are usually considered to study such distributions
Moreover, a combination of these summary statistics 
%the empty space and nearest neighbour distance distribution functions 
is provided through the (inhomogeneous) $J$-function \citep{VB96,van11}, which is a powerful quantifyer of points' tendency to cluster around or to inhibit each other. Although these summary statistics are well studied for spatial, spatio-temporal and marked point processes in $\R^d$, $d\geq1$  \citep{VB96,van2006Mark,van11,Cronie2015,Cronie2016},  their
linear network point process versions have not yet appeared in the literature. 
The reasons for this seem mainly to be related to theoretical challenges connected to the geometry of the linear network under consideration. 
%have mainly been
%they still did not get enough attention for the case where data points live on a linear network due to geometrical complexities and methodological problems. 
In this paper we tackle this problem and propose geometrically corrected analogues of these summary statistics for point processes on linear networks, which are defined based on regular linear network distances \citep{rakshit2017second}. Moreover, to do so we introduce the class of intensity reweighted moment pseudostationary (IRMPS) point processes, which in turn (perhaps less interestingly) yields a definition of stationarity for linear network point processes. 
To best connect our work to the existing literature on statistics for linear network point patterns, we carry out our numerical evaluations using the shortest-path distance as metric.

% In our numerical 

% Both the summary statistics and IRMPS
% More specifically, we follow \citet{rakshit2017second} and consider so-called regular linear networks distances when we define these entities.
% All 

% We define such summary statistics based on regular distances, however in our numerical evaluation and data analysis we make use of shortest-path distance. 

% , we focus on higher-order summary statistics and introduce intensity reweighted moment pseudostationary point processes on linear networks. For such point processes we then propose geometrically corrected analogues of the (inhomogeneous) empty space function and the (inhomogeneous) nearest neighbour distance distribution function and the $J$-function. 

% to make them accessible when studying network point data. We define such summary statistics based on regular distances, however in our numerical evaluation and data analysis we make use of shortest-path distance. The geometrically corrected higher-order summary statistics have known values for Poisson processes. We also discuss their non-parametric estimators as well as their statistical properties.

%\todo[inline]{Finish the introduction}

The paper is organised as follows. Section \ref{sec:Backgroun} provides a wide background of spatial point processes on linear networks. In Section \ref{sec:method} we review higher-order summary statistics and propose their geometrically corrected analogues for point patterns on linear networks. Section \ref{sec:numerical} is devoted to evaluating the performance of the geometrically corrected inhomogeneous linear $J$-function for a few models with different types of interaction. In Section \ref{sec:data} we apply the geometrically corrected inhomogeneous linear $J$-function to two real datasets. The paper ends with a discussion in Section \ref{sec:discuss}.

\section{Preliminaries
}\label{sec:Backgroun}
%We provide some background and definitions for linear networks as well as point processes on such structure. 
Throughout, $\R^d$, $d\geq1$, denotes the $d$-dimensional Euclidean space, $\|\cdot\|$ denotes the $d$-dimensional Euclidean norm, and all subsets under consideration will be Borel sets in the space in question. Moreover, $\int\de_1u$ will be used to denote integration with respect to arc length 
%(of some curve) 
and $\int\de x$ will be used to denote integration with respect to Lebesgue measure.

\subsection{Linear networks}\label{s:LinearNetworks}
%{\color{red}Below we should check that our notation etc is in keeping with \citet{rakshit2017second}}

Linear networks are, among other things, convenient tools for approximating geometric graphs/spatial networks. The spatial statistical literature usually defines a linear network as a finite union of (non-disjoint) line segments \citep{ABN12,BRT15,rakshit2017second}. % but, for different reasons, we here extend the definition to include unions of countably infinite numbers of line segments. 
More specifically, we define a linear network as a union
\[
L=\bigcup_{i=1}^k l_i, %= \bigcup_{i=1}^k [u_i,v_i] = \bigcup_{i=1}^k \{tu_i + (1-t)v_i:0\leq t\leq 1\}\subseteq\R^2
\qquad 1\leq k<\infty,
\]
of $k$ line segments $l_i=[u_i,v_i]=\{tu_i + (1-t)v_i:0\leq t\leq 1\}\subseteq\R^2$, $u_i\neq v_i\in\R^2$, with (arc) lengths $|l_i|=\|u_i-v_i\|\in(0,\infty)$, $i=1,\ldots,k$, which are such that any intersection $l_i\cap l_j$, $j\neq i$, is either empty or given by line segment end points. 
%We have that $k=0$ represents $L=\emptyset$. 
% To exclude pathological cases, we will also assume that $L$ is locally finite, by which we mean that any compact $A\subseteq\R^2$ intersects at most a finite number of line segments. 
We here restrict ourselves to connected networks since disconnected ones may simply be represented as unions of connected ones.

The Borel sets on $L$ are given by 
 $\B(L)=\{A\cap L:A\subseteq\R^2\}$ 
% $\B(L)=\{A\cap L:A\subseteq\B(\R^2)\}$ 
and they coincide with the $\sigma$-algebra generated by $\tau_L=\{A\cap L:A\text{ is an open subset of }\R^2\}$; recall that $A\subseteq L$ will mean that $A$ belongs to $\B(L)$. 
We further endow $L$ with the Borel measure $|A|=\nu_L(A)=\int_A\de_1u$, $A\subseteq L$, which represents integration with respect to arc length. Note that the total network length is given by $|L|=\sum_{i=1}^k |l_i|$.

\begin{remark}
%Since $L$ is a closed subset of the Polish space $\R^2$, it is itself a Polish space (under the subspace topology $\tau_L$).  %Moreover, it may be noted that $|\cdot|$ coincides with the $1$-dimensional Hausdorff measure on $\R^2$.
One could, in principle, also allow $k=\infty$ with the additional assumption of local finiteness, i.e.~any compact $A\subseteq\R^2$ intersects at most a finite number of line segments, which excludes pathological cases. This would result in the total network length $|L|=\infty$ and, as a consequence, one would allow networks which are isometric to $\R$.
\end{remark}

%\subsubsection{Graph representation}
Each linear network $L$ also has a graph theoretic interpretation. The endpoints of each line segment are called nodes/vertices, and the degree of each node is the number of line segments (edges) which share that node. The boundary of $L$ is the set of all nodes with degree one and is denoted by $\partial L$. See e.g.~\citet{eckardt2018point} for further details on graph theoretical aspects of linear networks.

\subsection{Linear network point processes}
Heuristically, a point process is a generalised sample in which the points may be dependent and the total point count may be random. More formally, given some probability space $(\Omega,\mathcal{F},\P)$, a (finite simple) point process $X=\{x_i\}_{i=1}^N$, $0\leq N<\infty$, on a linear network $L$ is a random element/variable in the measurable space $N_{f}$ of point configurations $\x=\{x_1,\ldots,x_n\}\subseteq L$, $0\leq n<\infty$; the associated $\sigma$-algebra is generated by the cardinality mappings $\x\mapsto N(\x\cap A)\in\{0,1,\ldots\}$, $A\subseteq L$, $\x\in N_{f}$, and coincides with the Borel $\sigma$-algebra generated by a certain metric on $N_{f}$ \citep{DVJ2}.

\subsubsection{Product densities}
Throughout, we will assume that the product densities/intensity functions $\rho^{(m)}$ of all orders $m\geq1$ exist. Formally, they may be defined through Campbell formulas: 
for any non-negative 
%Borel 
measurable function  $f(\cdot)$ on the product space $L^m$,
\begin{align}\label{eq:product}
\E\left[\mathop{\sum\nolimits\sp{\ne}}_{x_1,\ldots,x_m\in X}f(x_1,\ldots,x_m)\right]
&=
%\int_L\cdots\int_L 
\int_{L^m}
f(u_1,\ldots,u_m) \rho^{(m)}(u_1,\ldots,u_m)\de_1u_1\cdots\de_1u_m
.
\end{align}
%such that $\int_L\cdots\int_L |f(u_1,\ldots,u_m)|  \rho^{(m)}(u_1,\ldots,u_m)\de_1u_1\cdots\de_1u_m < \infty$ 
Here the notation $\sum^{\neq}$ is used to indicate that the summation is taken over distinct $m$-tuples. 
Since $X$ is simple, i.e.~$x_i\neq x_j$ for any $i\neq j$, $x_i,x_j\in X$, we interpret $\rho^{(m)}(u_1,\ldots,u_m)\de_1u_1\cdots\de_1u_m$ as the probability of jointly finding points of $X$ in some infinitesimal disjoint neighbourhoods $du_1,\ldots,du_m\subseteq L$ of $u_1,\ldots,u_m\in L$, with sizes $|du_1|=\de_1u_1,\ldots,|du_m|=\de_1u_m$.

In the particular case $m=1$, the right hand side of equation \eqref{eq:product} reduces to $\int_L f(u) \rho(u)\de_1u$, and in particular $\E[N(X\cap A)]=\int_A\rho(u)\de_1u$, $A\subseteq L$, where $\rho(u)=\rho^{(1)}(u)$ is called the intensity function of $X$. Whenever $\rho(u)=\rho>0$, $u\in L$, is constant, we say that $X$ is homogeneous and otherwise $X$ is called inhomogeneous.

\subsubsection{Correlation functions}
As with any joint probability structure and its relationship to its marginal probabilities, product densities are such that large/small values of $\rho^{(m)}(u_1,\ldots,u_m)$ do not necessarily imply that there is strong/weak dependence between points of $X$ located around $u_1,\ldots,u_m\in L$. For instance, for Poisson processes, where the points are independent, we have $\rho^{(m)}(u_1,\ldots,u_m)=\prod_{i=1}^m\rho(u_i)$ so any $\rho(u_i)$ being large may imply that $\rho^{(m)}(u_1,\ldots,u_m)$ is large. 
Instead, in order to study $m$-point dependencies it is more natural to consider so-called correlation functions $g_m$, $m\geq1$ (which do not actually represent correlations):
\begin{equation}
\label{CorrelationFunctions}
g_m(u_1,\ldots,u_m)
=
\frac{\rho^{(m)}(u_1,\ldots,u_m)}
{\rho(u_1)\cdots\rho(u_m)}, 
\qquad u_1,\ldots,u_m \in L.
\end{equation}
Note that $g_1(\cdot)=\rho(\cdot)/\rho(\cdot)=1$. 
Clearly, for a Poisson process with intensity function $\rho(\cdot)$ we have $g_m(\cdot)=1$, $m\geq1$, so we interpret $g_m(u_1,\ldots,u_m)>1$ as clustering/attraction between points of $X$ located around $u_1,\ldots,u_m\in L$. Similarly, $g_m(u_1,\ldots,u_m)<1$ indicates inhibition/regularity. There further exist recursively defined expansions of $g_m$, $m\geq1$ \citep{van11}:
\begin{align}
\label{Expansion}
    % \xi_1(u)&=1, %\qquad u\in L,\nonumber
    % \\
    g_m(u_1,\ldots,u_m)&=\sum_{j=1}^m \sum_{D_1,\ldots,D_j}\xi_{N(D_1)}\left(\{u_j:j\in D_1\}\right)\cdots\xi_{N(D_j)}\left(\{u_j:j\in D_j\}\right),
    %\quad m\geq1,
\end{align}
where the sum $\sum_{D_1,\ldots,D_j}$ ranges over all partitions $\{D_1,\ldots,D_j\}$ of $\{1,\ldots,m\}$ into $j$ non-empty and disjoint sets. 
For instance, $g_2(u,v)-g_1(u)=(\xi_2(u,v) + 1)-1=\xi_2(u,v)$. 
% {\color{blue} Multiplying both sides of \eqref{Expansion} by $\rho(u_1)\cdots\rho(u_m)$, which is the $m$-th order product density of a Poisson process with the same intensity function $\rho(\cdot)$ as $X$, we obtain an interpretation of the right hand side of \eqref{Expansion} as a dependence correction factor which we have to multiply the Poisson process product density with to obtain the product density $g_m(u_1,\ldots,u_m)$ of $X$.}

\subsubsection{Reduced Palm distributions}
A central tool in the study of a point process $X$ is its family of reduced Palm distributions $\{\P_u^!(X\in\cdot):u\in L\}$. Heuristically, $\P_u^!(X\in\cdot)$ represents the distribution of $X$ conditionally on $X$ having a point at $u$ which is removed once the process is realised; there actually exists a well-defined point process $X_u^!$ with distribution $\P_u^!(X\in\cdot)$. Formally, the most convenient way of defining $\{\P_u^!(X\in\cdot):u\in L\}$ is as the family of regular conditional distributions satisfying the reduced Campbell-Mecke formula \citep[Appendix C]{MW04}: For any non-negative and measurable mapping $f$ on $L\times N_{f}$,
\begin{align}
\label{RedCM}
    \E\left[\sum_{x\in X}f(x,X\setminus\{x\})\right]
    =
    \int_L\E[f(u,X_u^!)]\rho(u)\de_1u
    =
    \int_L\E_u^![f(u,X)]\rho(u)\de_1u,
\end{align}
where $\E_u^![\cdot]$ denotes expectation under $\P_u^!(X\in\cdot)$. 
\subsection{Second-order summary statistics}\label{s:SecondOrder}

Recalling \eqref{CorrelationFunctions}, 
the particular function 
\begin{eqnarray}\label{eq:pairnetwork}
g(u,v)=g_2(u,v)=\frac{\rho^{(2)}(u,v)}{\rho(u)\rho(v)}, \qquad u,v \in L,
\end{eqnarray}
which quantifies pairwise interactions in $X$, is commonly referred to as the pair correlation function. 
% As mentioned above, pair correlation functions characterise pairwise dependencies in point processes. 
In practice, however, it is often more convenient to work with cumulative versions, so-called $K$-functions. Statistical estimators of 
%Due to having knonw values for Poisson process, 
such functions may be considered in e.g.~exploratory data analyses, hypothesis testing and residual analyses. Further details on $K$-functions and pair correlation functions, together with their estimators, can be found in \citet{ABN12}, \citet[Chapter 17]{BRT15} and \citet{rakshit2017second}.

If we were to define the inhomogeneous $K$-function 
%$K_{\rm inhom}^L(r)$, $r\geq0$, 
of $X$ on $L$ in accordance with its original definition \citep{InhomK2000}, we would define it as 
\begin{align}
\label{KinhomBMW}
\bar K_{\rm inhom}^L(r)
&=
\frac{1}{|W|}\E\left[\mathop{\sum\nolimits\sp{\ne}}_{x_1,x_2\in X}\frac{\1\{x_1\in W\} \1\{x_2\in b_L(x_1,r)\}}{\rho(x_1)\rho(x_2)}\right]
\\
&=
\frac{1}{|W|}
\int_W
\E_u^!\left[\sum_{x\in X}\frac{\1\{x\in b_L(u,r)\}}{\rho(x)}\right]\de_1u, \nonumber
\end{align}
for $r\geq0$, some $W\subseteq L$, $|W|>0$, and the $r$-ball $b_L(u,r)=\{v\in L:d_L(u,v)\leq r\}$, $u\in L$, which is determined by a distance/metric $d_L$ on $L$; 
the second equality follows by \eqref{RedCM}. 
There are, however, a few questions which immediately appear here. 

The first question is related to something as basic as what the ball $b_L(u,r)$ looks like. 
Note further that the geometry of the ball $b_L(u,r)$ may change with $u\in L$ (there may e.g.~exist $u,v\in L$, $u\neq v$, such that $|b_L(u,r)|\neq|b_L(v,r)|$). 
Secondly, for the Euclidean version $\bar K_{\rm inhom}^{\R^d}(r)$, which is obtained by replacing $W\subseteq L$ by $W\subseteq \R^d$ and letting $X$ be a point process on $\R^d$ in \eqref{KinhomBMW}, \citet{InhomK2000} assumed so-called second-order intensity reweighted stationarity (SOIRS), i.e.~that the intensity function of $X\subseteq\R^d$ is strictly positive and that the pair correlation function of $X$ is translation invariant in the sense that $g(u_1,u_2)=\bar g(u_1-u_2)$, $u_1,u_2\in\R^d$, for some function $\bar g:\R^d\to[0,\infty)$. 
Note that in the linear network setting, SOIRS cannot be defined in an analogous way since the vector $u_1-u_2$, and thereby $\bar g(u_1-u_2)$, does not make sense on a linear network (e.g., for most $u_1,u_2\in L$ we have $u_1-u_2\notin L$).
%, in order to obtain that the reduced Palm expectation integrands above do not depend on $u\in\R^d$. 
In the Euclidean setting, under SOIRS and by the Campbell formula, we obtain that 
\begin{align}
\label{KinhomBMWEuclidean}
\bar K_{\rm inhom}^{\R^d}(r)
=
\frac{1}{|W|}
\int_W\int_{b_{\R^d}(u_1,r)} g(u_1,u_2) \de u_1\de u_2
=
\int_{b_{\R^d}(u,r)} \bar g(u) \de u,
\end{align}
for any $u\in\R^d$. Hence, $\bar K_{\rm inhom}^{\R^d}(r)$ does not depend on the choice of the  (positively Lebesgue sized) set $W\subseteq\R^2$ in \eqref{KinhomBMW} and the reduced Palm expectation in the Euclidean version of \eqref{KinhomBMW} does not depend on $u\in\R^d$.
This last observation is highly important in statistical settings -- essentially, we need the reduced Palm integrand to be constant in $u$ since we then may carry out the estimation of the inhomogeneous $K$-function by averaging over points in $X\cap W$, rather than requiring repeated samples of $X\cap W$.

It turns out that the choice of metric $d_L$ and an appropriate notion of SOIRS for linear network point processes are related to each other, and that the former gives rise to the latter. We next look closer at metric choices. 

\subsubsection{General regular distance metrics}
\citet{ABN12} used the shortest-path distance as metric $d_L$ and defined a linear network point process as being {\em second-order intensity reweighted pseudostationary (SOIRPS)} if $g(u_1,u_2)=\bar g(d_L(u_1,u_2))$, $u_1,u_2\in L$, for some function $\bar g:[0,\infty)\to[0,\infty)$; note that \citet{ABN12} referred to this as second-order reweighted pseudostationary (SORS). 
Taking the work of \citet{ABN12} a step further, \citet{rakshit2017second} extended the second-order analysis for linear network point processes to incorporate a larger class of metrics on linear networks, namely so-called {\em regular distance metrics}; the shortest-path distance belongs to this family. The motivation for doing so had partly to do with there being several metrics available which suit different applications and partly to do with the family of SOIRPS point processes being relatively small. 

In accordance with \citet{rakshit2017second}, we next give a brief account on regular distance metrics $d_L$ -- note that we make no explicit assumption that $d_L$ metrises $L$, i.e.~that $d_L$ generates the (subspace) topology $\tau_L$ inherited from $\R^2$ (see Section \ref{s:LinearNetworks}). Recall that a (distance) metric on $L$ is a function $d_L: L \times L \to [0,\infty)$ satisfying i) $d_L(u,u)=0$, ii) $d_L(u,v)=d_L(v,u)$, and iii) $d_L(u,v)\leq d_L(u,w)+d_L(v,w)$ for any $u,v,w\in L$. 
%{\color{red}Write about differentiability.}

\begin{definition}\citep[Definition 1]{rakshit2017second}
A {\em regular distance metric} on a linear network $L$ is a metric $d_L: L \times L \to [0,\infty)$ which further satisfies i) being a continuous function of $u$ and $v$, and ii) for any fixed $u \in L$, the partial derivative $\partial d_L(u,v)/\partial v$ exists and is nonzero everywhere except at a finite set of locations $v$.
\end{definition}

We next comment on integration over linear networks. 
Fix a point $u\in L$. 
For an integrable real-valued function $f$, 
we have the following change-of-variables formula \citep[Proposition 1]{rakshit2017second}: 
\begin{eqnarray}\label{eq:intregulardist}
\int_L f(v) \de_1 v= \int_0^{\infty} \sum_{v \in L: d_L(u,v)=r} \frac{f(u)}{J_{d_L}(u,v)} \de r,
\end{eqnarray}
where $J_{d_L}(u,v)= \left| \partial d_L(u,v)/\partial v \right|$ is the Jacobian; there may be a finite collection of fixed points $u\in L$ such that \eqref{eq:intregulardist} does not hold. If there is no $u \in L$ such that $d_L(v,u)=r$, the sum on the right hand side of equation \eqref{eq:intregulardist} is 0. Extending equation \eqref{eq:intregulardist} to functions $f: L^m \to \R$, $m\geq1$, we have
\begin{align}\label{eq:intregulardistmult}
&\int_L \cdots \int_L f(u_1,\ldots,u_m) \de_1 u_1 \cdots \de_1 u_m =
\int_0^{\infty} \sum_{u_1 \in L:d_L(u_1,u)=r_1} \frac{1}{J_{d_L}(u_1,u)} \cdots
\\
&
\cdots
\int_0^{\infty} \sum_{u_m \in L:d_L(u_m,u)=r_m}
\frac{1}{J_{d_L}(u_m,u)} f(u_1,\ldots,u_m) \de r_1 \cdots \de r_m\nonumber
\\
&
=
\int_0^{\infty} \cdots
\int_0^{\infty}
\sum_{u_1 \in L:d_L(u_1,u)=r_1}
\cdots
 \sum_{u_m \in L:d_L(u_m,u)=r_m}
 \frac{f(u_1,\ldots,u_m)}{\prod_{i=1}^mJ_{d_L}(u_i,u)} \de r_1 \cdots \de r_m.\nonumber
\end{align}

Let $\mathcal{D}(u)= \max \{ d_L(u,v): v \in L \}$ be the farthest reachable distance from $u$ and $c_L(u,r)$ be the number of points exactly $r$ units away from $u$ according to the choice $d_L$. Then, for $ r\in(0,\mathcal{D}(u))$, consider 
\begin{eqnarray}
\tilde{J}_{d_L}(u,r)= \left[ \frac{1}{c_L(u,r)} \sum_{v \in L: d_L(v,u)=r} \frac{1}{J_{d_L}(u,v)} \right]^{-1},
\end{eqnarray}
which is the harmonic mean of $J_{d_L}(u,v)$ at all locations $v$ exactly $r$ units away from $u$ according to $d_L$. Moreover, the harmonic mean is zero if any $J_{d_L}(u,v)$ is zero. 
Defining 
$$
w_{d_L}(u,r)=
\frac{\tilde{J}_{d_L}(u,r)}{c_L(u,r)},
$$
and if the function $f$ in \eqref{eq:intregulardist} only depends on the distance $d_L(u,v)$, we obtain
\begin{eqnarray}\label{eq:intregulardistpairs}
\int_L f(v) \de_1 v = \int_L h(d_L(u,v)) w_{d_L}(u,d_L(u,v)) \de_1 v 
=
\int_0^{\mathcal{D}(u)} h(r) \de r,
\end{eqnarray}
where $h:[0, \infty) \to \R$. Hence, the equation above  can be extended to 
\begin{align}\label{eq:intregulardistpairsmult}
&\int_L \cdots \int_L h(d_L(u_1,u),\ldots,d_L(u_m,u)) \prod_{i=1}^{m} w_{d_L}(u,d_L(u_i,u)) \de_1 u_1 \cdots \de_1 u_m =
\\
&
=
\int_0^{\mathcal{D}(u)} 
\cdots
\int_0^{\mathcal{D}(u)} h(r_1,\ldots,r_m) \de r_1 \cdots \de r_m \nonumber
\end{align}
for $h:[0, \infty)^m \to \R$.

\subsubsection{Second-order intensity reweighted pseudostationarity and inhomogeneous linear network $K$-functions}

\citet{rakshit2017second} considered point processes for which the pair correlation function \eqref{eq:pairnetwork} only depends on the regular distance metric $d_L$, i.e.~$g(u,v)=\bar{g}(d_L(u,v))$ for some function $\bar g:[0,\infty)\to[0,\infty)$, and called such a point process {\em $d_L$-correlated}; with their notation,  {\em $\delta$-correlated} since they used the notation $\delta$ for $d_L$. This is exactly what \cite{ABN12} called {\em second-order reweighted pseudostationary} when $d_L$ is given by the shortest-path distance. For any inhomogeneous $d_L$-correlated point process $X$ and any $0\leq r < R$, where $R=\min_{u \in L} \mathcal{D}(u)$, \citet{rakshit2017second} proposed the following geometrically corrected inhomogeneous $K$-function
\begin{align}
\label{KinhomAng}
K_{\rm inhom}^L(r)
&=
\E_u^!\left[\sum_{x\in X}\frac{\1\{x\in b_L(u,r)\}}{\rho(x)}
w_{d_L}(u,d_L(u,x))\right]
\nonumber
\\
&=
\frac{1}{|L|}\E\left[\mathop{\sum\nolimits\sp{\ne}}_{x_1,x_2\in X}\frac{\1\{x_1\in W\} \1\{x_2\in b_L(x_1,r)\}}{\rho(x_1)\rho(x_2)}
w_{d_L}(x_1,d_L(x_1,x_2))
\right] \nonumber
\\
&=
\int_0^{r} \bar g(t)\de t,
\end{align}
%{\color{blue} What is W?}
for any $u\in L$. Note that here $\bar g: [0,\infty)\to [0,\infty)$. For a Poisson process on the linear network $L$, which has pair correlation $g(\cdot)=1$, 
we have $K(r)=r$; values larger than $r$ indicate clustering within (pairwise) inter-point distance $r$ while values smaller than $r$ instead reveal inhibition. Similarly, $g(r)>1$ indicates a clustering behaviour between $r$-separated point pairs while $g(r)<1$ points to inhibition.

\section{
Higher-order summary statistics
%Methodology
}\label{sec:method}
To quantify degrees of dependence of any order higher than two between points in a stationary point process in $\R^d$, three common and powerful tools are given by the empty space function, the nearest neighbour distance distribution function and the $J$-function 
%, which were originally considered in 
\citep{bartlett1964spectral,paloheimo1971theory,diggle79, VB96}.
%{\color{red}(Four references but only three summary statistics! Is the Bartlett paper actually considering any of these?)}
In a seminal paper, \cite{van11}  finally extended these summary statistics to inhomogeneous point processes in $\R^d$ and later \citet{Cronie2015,Cronie2016} extended them further to spatio-temporal and marked inhomogeneous point processes in $\R^d$. 
Our aim here is to study these higher order summary statistics and their non-parametric estimation in the context of linear network point processes. 
%{\color{red} However, such extensions are all but straightforward mathematically, and to be able to proceed we will have to make some pragmatic assumptions.} 

\subsection{The general case}\label{sec:GeneralSumStat}

Given some arbitrary (complete and separable metric) space $S$ and some locally finite Borel reference measure $\nu_S$, consider a point process $X$ in $S$ with product densities $\rho^{(m)}$ (with respect to products of $\nu_S$), $m\geq1$, and $\bar\rho=\inf_{u\in S}\rho(u)>0$. 
%Given a point process $X$ on some linear network $L$ with $\bar\rho=\inf_{u\in L}\rho(u)>0$, 
The summary statistics of \citet{van11} may now be expressed as follows. 
Given the closed ball $b_S(u,r)=\{v\in S:d_S(u,v)\leq r\}\subseteq S$ with centre $u\in S$ and radius $r\geq0$, which is based on some metric $d_S(\cdot,\cdot)$ on $S$,
let 
\begin{align}
\label{AbsConv}
\limsup_{m\to\infty}
\left(
\frac{\bar\rho^m}{m!}
\int_{b_S(u,r)^m}
g_m(u_1,\ldots,u_m)
\nu_S(du_1)\cdots\nu_S(du_n)
\right)^{1/m} < 1
\end{align}
and consider the \emph{inhomogeneous empty space function} (at $u\in S$) 
\begin{align}
\label{Finhom}
F_{\rm inhom}^S(r;u) 
&= 
1 - \E\left[
\prod_{x \in X} \left(
1 -
\frac{\bar\rho}{\rho(x)}
\1\{x\in b_S(u,r)\}
\right)
\right]
% \\
% &=
% \sum_{m=1}^{\infty}
% \frac{(-\bar\rho)^m}{m!}
% \int_{b_S(u,r)^m}
% g_m(u_1,\ldots,u_m)
% \nu_S(du_1)\cdots\nu_S(du_n)
% ,
% %\qquad r\geq0, %u\in S
,
% \nonumber
\end{align}
the {\em inhomogeneous nearest neighbour distance distribution function} (at $u\in S$) %is given by
\begin{align}
\label{Hinhom}
H_{\rm inhom}^S(r;u) 
&=
1 - \E_u^!\left[
\prod_{x \in X} \left(
1 -
\frac{\bar\rho}{\rho(x)}
\1\{x\in b_S(u,r)\}
\right)
\right]
% \\
% &=
% \sum_{m=1}^{\infty}
% \frac{(-\bar\rho)^m}{m!}
% \E_u^!\left[
% \mathop{\sum\nolimits\sp{\ne}}_{x_1,\ldots,x_m\in X}
% \prod_{i=1}^m
% \frac{\1\{x_i\in b_S(u,r)\}}{\rho(x_i)}
% \right]
% \nonumber
% \\
% &=
% \sum_{m=1}^{\infty}
% \frac{(-\bar\rho)^m}{m!}
% \int_{b_S(u,r)^m}g_{m+1}(u,u_1,\ldots,u_m)\de_1u_1\cdots\de_1u_m
,
%\qquad r\geq0, %u\in S,
%\nonumber
\end{align}
and the {\em inhomogeneous $J$-function} (at $u\in S$) %is given by
\begin{align}
\label{Jinhom}
J_{\rm inhom}^S(r;u) 
&= 1 + \sum_{m=1}^{\infty}\frac{(-\bar\rho)^m}{m!}
\int_{b_S(u,r)^m}
%\cdots\int_{b_S(u,r)}
\xi_{m+1}(u,u_1,\ldots,u_m)\nu_S(du_1)\cdots\nu_S(du_m).
\end{align}

Given $X$ within some bounded $W\subseteq S$, non-parametric estimators of \eqref{Finhom} and \eqref{Hinhom} based on $X\cap W$ are given by
\begin{align}
\label{FinhomLocal}
\widehat F_{\rm inhom}^S(r;u,X,W) 
&=
\prod_{x \in X\cap W} \left(
1 -
\frac{\bar\rho}{\rho(x)}
\1\{x\in b_S(u,r)\}
\right)
,
\quad u\in W,
\\
\label{HinhomLocal}
\widehat H_{\rm inhom}^S(r;u,X,W) 
&=
\prod_{x \in X\cap W\setminus\{u\}} \left(
1 -
\frac{\bar\rho}{\rho(x)}
\1\{x\in b_S(u,r)\}
\right)
,
\quad u\in X\cap W,
\end{align}
respectively; note that in practice we would plug in an estimate of $\rho(\cdot)$ into these estimators. These are local empirical summary statistics.

\subsection{The Euclidean case}
Besides providing the definitions of the summary statistics in \eqref{Finhom}, \eqref{Hinhom} and \eqref{Jinhom}, for which the intuition will be clarified in a moment, \citet[Theorem 1]{van11} showed that under certain conditions, when $S=\R^d$, $b_S(u,r)=b_{\R^d}(u,r)=\{x\in\R^d:\|x-u\|\leq r\}$, $u\in\R^d$, where $\|\cdot\|$ is the Euclidean norm, the functions in Section \ref{sec:GeneralSumStat} are almost everywhere constant as functions of $u\in S=\R^d$. We may thus write $H_{\rm inhom}^{\R^d}(r)$ and $F_{\rm inhom}^{\R^d}(r)$ for \eqref{Finhom} and  \eqref{Hinhom} to emphasize this indepencence of the choice of $u\in\R^d$. In addition, \citet[Theorem 1]{van11} also showed that \eqref{Jinhom} satisfies
\begin{align}\label{eq:Jinhom}
J_{\rm inhom}^{\R^d}(r;u)
=
J_{\rm inhom}^{\R^d}(r)
% =\frac{1-H_{\rm inhom}^{\R^d}(r;u)}{1-F_{\rm inhom}^{\R^d}(r;u)}
=
\frac{1-H_{\rm inhom}^{\R^d}(r)}{1-F_{\rm inhom}^{\R^d}(r)}
, \qquad r\geq0, %, F_{LI}<1,
\end{align}
for almost every $u\in \R^d$ and $F_{\rm inhom}^{\R^d}(r)\neq1$, and truncating the sum in \eqref{Jinhom} at $m=1$ here yields
\begin{align}
\label{Truncation}
J_{\rm inhom}^{\R^d}(r) 
&\approx 1 -\bar\rho
\int_{b_{\R^d}(o,r)}\xi_{2}(o,u)\de u
=1 - \bar\rho(\bar K_{\rm inhom}^{\R^d}(r) - |b_{\R^d}(o,r)|),
\end{align}
where e.g.~$o=(0,\ldots,0)\in\R^d$ and we recall $\bar K_{\rm inhom}^{\R^d}(r)$ from equation \eqref{KinhomBMWEuclidean}.

% $\bar K_{\rm inhom}^{\R^d}(r)$, the Euclidean version of \eqref{KinhomBMW}.  
% {\color{blue} Why not to refer to the eq after \eqref{KinhomBMW}}

The main importance of \eqref{eq:Jinhom} is that it is an inhomogeneous analogue/natural extension of the $J$-function of \citet{VB96} for stationary point processes:
\begin{align}\label{eq:J}
J^{\R^d}(r)=\frac{1-H^{\R^d}(r)}{1-F^{\R^d}(r)}
=\frac{1-\P^!_o(X \cap b_{\R^d}(o,r) \neq \emptyset)}{1-\P(X \cap b_{\R^d}(o,r) \neq \emptyset)}
=\frac{\P^!_o(X \cap b_{\R^d}(o,r) = \emptyset)}{\P(X \cap b_{\R^d}(o,r) = \emptyset)}
, \quad r\geq0, %F_L(r)<1,
\end{align}
%where $o\in S$ is some origin % (e.g.~$o=(0,\ldots,0)\in\R^d$) 
%and 
where we recall that under stationarity $\P^!_o$ is chosen to represent the family $\{\P^!_u(\cdot):u\in \R^d\}$ since $\P^!_u$ is constant as a function of $u\in\R^d$. We emphasise that under stationarity, \eqref{eq:Jinhom} reduces to \eqref{eq:J}; note that here $\rho(u)=\bar\rho=\rho>0$ for any $u\in\R^d$. 
Noting that for a Poisson process $X$ we have $\P^!_u(X\in\cdot)=\P(X\in\cdot)$ for any $u$ by Slivniyak's theorem \citep{CSKWM13}, we see that for Poisson processes we have 
%$H_{\rm inhom}^{\R^d}(r)=F_{\rm inhom}^{\R^d}(r)=$
$J^{\R^d}(r)=1$, $r\geq0$; also $J_{\rm inhom}^{\R^d}(r)=1$ holds true for (inhomogeneous) Poisson processes. The interpretation is thus that we quantify how conditioning on there being a point of $X$ at some location increases/decreases the probability of seeing a further point within distance $r$. Hence, $J_{\rm inhom}^{\R^d}(r)>1$ indicate inhibition/regularity whereas these quantities being smaller than 1 indicate clustering/attraction between points of $X$ with inter-point distance at most $r$ -- in the inhomogeneous case this should be understood in the sense of having scaled away the individual intensity contributions of the points of $X$. 

%\subsection{Different notions of stationarity and transformations on linear networks}
One thing that has completely been left out of the discussion above is a discussion on the definition of stationarity. 
% When dealing with point processes in $\R^d$, one often encounters the notion of stationarity. 
A point process $X=\{x_i\}_i^N$ in $\R^d$ being stationary means that its distribution satisfies $\P(\{x+y:x\in X\}\in\cdot)=\P(X\in\cdot)$ for any $y\in\R^d$.
% {\color{blue} Stationarity for a point process $X=\{x_i\}_i^N$ in $\R^d$ means that its distribution $P(\cdot)=\P(X\in\cdot)$ satisfies $\P(\{x+y:x\in X\}\in\cdot)=\P(X\in\cdot)$ for any $y\in\R^d$.} 
In other words, its distribution is invariant under a family $\mathcal{T}$ of transformations on the spatial domain, which here is given by the (Euclidean) family $\mathcal{T}=\mathcal{T}_{\R^d}=\{T_y:\R^d\to\R^d\}_{y\in\R^d}$, $T_y x = x+y$, $x\in\R^d$, of translations/shifts. 
Note that stationarity is a very strong assumption which, among other things, implies homogeneity, i.e.~that $X$ has a constant intensity. 
Moreover, we have mentioned that i) 
%the summary statistics 
\eqref{Finhom}, \eqref{Hinhom} and \eqref{Jinhom} being constant in $u\in S=\R^d$, and ii) the relation \eqref{eq:Jinhom} being satisfied, were proved under some conditions in \citet{van11}. 
% For the purpose of defining empty space functions, nearest neighbour distance distribution functions and $J$-functions for inhomogeneous point processes in $\R^d$, 
More specifically, what \citet{van11} assumed was so-called intensity reweighted moment stationarity (IRMS) for the point process $X$. We say that a point process $X=\{x_i\}_i^N$ in $\R^d$ with correlation functions $g_m$, $m\geq1$, is IRMS if i) $\bar\rho=\inf_{u\in\R^d}\rho(u)>0$, and ii) $g_m(u_1,\ldots,u_m) = g_m(T_yu_1,\ldots,T_yu_m) = g_m(u_1+y,\ldots,u_m+y)$ for almost any $u_1,\ldots,u_m\in\R^d$, any $y\in\R^d$ and any $m\geq1$. In fact, the original definition of \citet{van11} states, equivalently, that the expansion terms $\xi_m$, $m\geq1$, in \eqref{Expansion} should be translation invariant in the above sense. In other words, all intensity reweighted factorial moments should be invariant under the transformations $\mathcal{T}_{\R^d}$. Note further that stationarity implies IRMS.

The proof of 
Theorem 1 in \citet{van11} exploits that $X$ is IRMS to obtain that 
%the summands in 
\eqref{Finhom} and \eqref{Hinhom} are almost everywhere constant as functions of $u\in\R^d$. 
This in turn allows us to treat the distributions of \eqref{FinhomLocal} and \eqref{HinhomLocal} as constant with respect to $u$ and, as such, we may consider the estimators
\begin{align}
\label{EstRd}
\widehat F_{\rm inhom}^{\R^d}(r;X,W) 
&=
\frac{1}{N(I\cap W_{\ominus r})}
\sum_{u\in I\cap W_{\ominus r}}
\widehat F_{\rm inhom}^{\R^d}(r;u,X,W)
,
\\
\widehat H_{\rm inhom}^{\R^d}(r;X,W) 
&=
\frac{1}{N(X\cap W_{\ominus r})}
\sum_{u\in X\cap W_{\ominus r}}
\widehat H_{\rm inhom}^{\R^d}(r;u,X,W)
,\nonumber
\end{align}
where $I\subseteq W\subseteq \R^d$ is a fine grid and $W_{\ominus r}$ is the $r$-erosion of $W$.

% in the definitions above, \citet{van11} showed that when $X$ is IRMS and the correlation functions are sufficiently nice, all the definitions above are constant for (almost) every $u\in\R^d$ and 
% \begin{eqnarray}\label{eq:Jinhom}
% J_{LI}(r)=\frac{1-H_{LI}(r)}{1-F_{LI}(r)}, \qquad r\geq0. %, F_{LI}<1,
% \end{eqnarray}

%{\color{red}Interpretation of the sum stats}

\subsection{The linear network case}

In either of the two forms of distributional invariance appearing in the Euclidean setting above, we assume some form of distributional invariance with respect to a family of transformations. Hence, if we would like to consider different forms of distributional invariance for a point process $X$ on some linear network $L$, we would need to find a suitable family of transformations $\mathcal{T}$, given the distribution $P(\cdot)=\P(X\in\cdot)$. 
Formally, this means working in the setting of 
%Lie or locally compact 
algebraic groups/geometric measure theory -- to reproduce the proof of \citet[Theorem 1]{van11} in the linear network setting we would have to exploit Haar measure based arguments, where our reference measure $\nu_L$ would be a Haar measure and the associated collection of transformations, $\mathcal{T}$, would have a group structure and act (transitively) on $L$. 
%It is tempting to try to directly translate the ideas from $\R^d$ to the linear network context. 
However, such a structure seems (to the best of our knowledge) to not be available; cf.~\citet{baddeley2017stationary}. %and the reason seems to be the lack of so-called transient families of transformations on linear networks. 
As a consequence, the independence of $u$ in the summary statistics \eqref{Finhom}, \eqref{Hinhom} and \eqref{Jinhom} may not be attainable so when performing non-parametric estimation we may not be able to justify the type of averaging over points of $X\cap W$, $W\subseteq L$, that was considered in \eqref{EstRd}. This clearly poses a problem since in general we do not have access to repeated samples of $X\cap W$. 

Our solution to obtaining geometrically corrected summary statistics comes from combining expression \eqref{Truncation} with \eqref{KinhomBMW} and \eqref{KinhomAng}. More precisely, in the Euclidean setting the truncation \eqref{Truncation} contains the inhomogeneous $K$-function $\bar K_{\rm inhom}^{\R^d}(r)$ in \eqref{KinhomBMWEuclidean}, but taking the discussion in Section \ref{s:SecondOrder} into consideration, since we are dealing with linear networks we should instead have the geometrically corrected $K$-function $K_{\rm inhom}^L(r)$ in \eqref{KinhomAng} in the truncation. By looking closer at \eqref{Finhom}, \eqref{Hinhom} and \eqref{Jinhom}, and revisiting the results and proofs in \citet{van11}, we arrive at the definitions below. 
%{\color{red}Note that we have added the term ``pseudo" in front of the names of the summary statistics to indicate that they are not necessarily distribution functions.}

\begin{definition}\label{SumStatsL}
Given a point process $X$ on a linear network $L$, with product densities $\rho^{(m)}$, $m\geq1$, and $\bar\rho=\inf_{u\in L}\rho(u)>0$, the inhomogeneous geometrically corrected \emph{linear empty space function}, the {\em linear nearest neighbour distance distribution function} and the {\em linear $J$-function} at $u\in L$ and $r\geq0$, with respect to a regular distance metric $d_L$, are given by
\begin{align}
\label{FinhomL}
& F_{\rm inhom}^L(r;u) 
= 
1 - \E\left[
\prod_{x \in X} \left(
1 -
\frac{\bar\rho\1\{x\in b_L(u,r)\}w_{d_L}(u,d_L(u,x))}{\rho(x)}
\right)
\right]
\\
&=
-\sum_{m=1}^{\infty}
\frac{(-\bar\rho)^m}{m!}
\int_{b_L(u,r)^m}
g_m(u_1,\ldots,u_m)
\prod_{i=1}^m w_{d_L}(u,d_L(u,u_i))
\de_1u_1\cdots\de_1u_m,
\nonumber
\\
\label{HinhomL}
&H_{\rm inhom}^L(r;u) 
= 
1 - \E_u^!\left[
\prod_{x \in X} \left(
1 -
\frac{\bar\rho\1\{x\in b_L(u,r)\}w_{d_L}(u,d_L(u,x))}{\rho(x)}
\right)
\right] 
\\
&=
-\sum_{m=1}^{\infty}
\frac{(-\bar\rho)^m}{m!}
K_{\rm inhom}^{L,m}(r;u)\nonumber
%\nonumber
% \\
% &=
% \sum_{m=1}^{\infty}
% \frac{(-\bar\rho)^m}{m!}
,
% \nonumber
\\
\label{JinhomL}
&J_{\rm inhom}^L(r;u) 
= 1 + \sum_{m=1}^{\infty}\frac{(-\bar\rho)^m}{m!}
\times
\\
&\times
\int_{b_L(u,r)^m}%\cdots\int_{b_L(u,r)}
\xi_{m+1}(u,u_1,\ldots,u_m)
\prod_{i=1}^m w_{d_L}(u,d_L(u,u_i))
\de_1u_1\cdots\de_1u_m,
\nonumber
\end{align}
respectively, where
\begin{align*}
    K_{\rm inhom}^{L,m}(r;u)
    &=
    \E_u^!\left[
\mathop{\sum\nolimits\sp{\ne}}_{x_1,\ldots,x_m\in X}
\prod_{i=1}^m
\frac{\1\{x_i\in b_L(u,r)\}w_{d_L}(u,d_L(u,x_i))}{\rho(x_i)}
\right]
\\
&=
\int_{b_L(u,r)^m}
g_{m+1}(u,u_1,\ldots,u_m)
\prod_{i=1}^m w_{d_L}(u,d_L(u,u_i))
\de_1u_1\cdots\de_1u_m
;
\end{align*}
note that $m=1$ gives us \eqref{KinhomAng}.
\end{definition}

%{\color{red}Which are the ranges $r\geq0$ that we can consider? Do we need to assume that an $r$-ball fits into $L$ -- does this mean that we have to have $|L|=\infty$?}

The expansion in \eqref{FinhomL} follows from an application of the Campbell formula, and the expansions in \eqref{HinhomL} follow from a combination of the Campbell formula and the Campbell-Mecke formula \citep{van11,Cronie2015}.

The missing ingredient is still some form of distributional invariance, similar to IRMS in the Euclidean context, where the essence is that the correlation functions should only depend on the inter-point distances of the points. Intuitively, we would translate this idea to the linear network setting by letting the correlation functions only depend on the $d_L$-distances between the input points, i.e.
\begin{align}
\label{IRMPSweak}
g_m(u_1,\ldots,u_m)
=
\widetilde g_m(\{d_L(u_i,u_j): i,j=1,\ldots,m, i\neq j\}), \quad u_1,\ldots,u_m\in L, %1\leq m\leq k,
\end{align}
for some family of functions $\widetilde g_m$, $m\geq1$; requiring this to hold for (only/at least) $m\leq 2$ would essentially yield the definition of second-order reweighted pseudostationarity in \citet{rakshit2017second}. 
It turns out that this is not completely sufficient and we have to impose the slightly stronger condition that $g_m(u_1,\ldots,u_m)$ is given by a function of the distances $d_L(u,u_i)$, $i=1,\ldots,m$, where $u\in L$ is arbitrary.
Note that if the metric $d_L$ is independent of a chosen origin, the two concepts coincide. 
%-- check Section 5 in \citet{rakshit2017second} to see if they have given any such examples; \citet{MollerCovariance} says that the resistance metric has this property I think).
By additionally assuming homogeneity here we obtain an expression pertaining to the product densities and thereby a definition of moment stationarity for linear network point processes. If we here also assume that the moments characterise the whole distribution of the point process we obtain a definition of (pseudo)stationarity for linear network point processes. Note that we have chosen the names below in keeping with \citet{van11}  and \citet{ABN12}.

\begin{definition}\label{IRMPS}
Let $X$ be a point process on a linear network $L$ and let $d_L$ be a regular distance metric on $L$. 
Given some $k\geq2$, whenever the product densities $\rho^{(m)}$, $1\leq m\leq k$, exist, $\bar\rho=\inf_{u\in L}\rho(u)>0$ and for any $m\in\{2,\ldots,k\}$ the correlation function $g_m:L^m\to[0,\infty)$ satisfies
\begin{align}
    \label{IRMPSstrong}
g_m(u_1,\ldots,u_m)
=
\bar g_m(d_L(u,u_1),\ldots,d_L(u,u_m))
\end{align}
for any $u\in L$ and some function
$\bar g_m:[0,\infty)^m\to[0,\infty)$, 
we say that $X$ is {\em $k$-th order intensity reweighted pseudostationary (with respect to $d_L$)}; 
when this holds for any order $k\geq2$ we say that $X$ is {\em intensity reweighted moment pseudostationary (IRMPS)}.
% As a slightly stronger alternative one could assume that 
% .

If $X$ is both homogeneous with $\rho(u)=\bar\rho=\rho>0$, $u\in L$, and $k$-th order intensity reweighted pseudostationary, we say that it is {\em $k$-th order pseudostationary}; note that 
\(
\rho^{(m)}(u_1,\ldots,u_m)
=
\bar\rho^{(m)}(d_L(u,u_1),\ldots,d_L(u,u_m)) 
\)
for any $u\in L$ here. When this holds for any $k\geq1$ we say that $X$ is {\em moment pseudostationary}. 
Finally, any moment pseudostationary $X$ such that the (factorial) moments completely and uniquely characterise its distribution may be referred to as (strongly) {\em pseudostationary}.
\end{definition}

Note that 
% the case $k=1$ is void since it is always satisfied. Moreover, 
the last condition in the definition above, i.e.~that the moments completely and uniquely characterise the distribution of $X$, may be obtained by requiring that the conditions in \citet[Section 2]{Zessin} hold. It is worth mentioning that (to the best of our knowledge) our definition of pseudostationarity is the first version of some form of ``strong stationarity" for linear network point processes to be provided in the literature. 

%To see that \eqref{IRMPSstrong} implies \eqref{IRMPSweak}, note that if we are standing on one of $\{u_1,\ldots,u_m\}$ and know the distances to all of the remaining points, we obtain one row/column of the inter-point distance matrix. Moving through all of the points we fill the whole matrix. 

Next, we state our main result, which is essentially a (geometrically corrected) linear network version of \citet[Theorem 1]{van11}; its proof can be found in the Appendix.

\begin{thm}\label{ThmRepresentation}
For any IRMPS point process $X$ on a linear network $L$, which also satisfies \eqref{AbsConv} with $S=L$, the summary statistics in Definition \ref{SumStatsL} satisfy  
%Moreover, if $X$ is also pseudostationary then 
\begin{align*}
    F_{\rm inhom}^L(r;u)&=F_{\rm inhom}^L(r)
= -\sum_{m=1}^{\infty}
\frac{(-\bar{\rho})^m}{m!}
\int_0^r \cdots \int_0^r \bar g_{m}(t_1,\ldots,t_m)\de t_1 \cdots \de t_m
,
\\
H_{\rm inhom}^L(r;u)
&= H_{\rm inhom}^L(r)= -\sum_{m=1}^{\infty} \frac{(-\bar{\rho})^m}{m!}
\int_0^{r}
\cdots
\int_0^{r}
\bar g_{m+1}(0,t_1,\ldots,t_m)
\de t_1\cdots\de t_m
,
\\
J_{\rm inhom}^L(r;u)
&= 
J_{\rm inhom}^L(r)
=
\frac{1-H_{\rm inhom}^L(r)}{1-F_{\rm inhom}^L(r)},
\end{align*}
for (almost) any $u\in L$. 
\end{thm}

When $X$ is pseudostationary we immediately obtain the following corollary, upon noting that the intensity is constant.
%Letting $X$ in Theorem \ref{ThmRepresentation}
\begin{cor}
Let $X$ be pseudostationary with constant intensity $\rho>0$. Then we obtain %the versions 
(cf.~\citet[p.~186]{van11})
\begin{align*}
    F_{\rm inhom}^L(r)&=F^L(r)
= -\sum_{m=1}^{\infty}
\frac{(-1)^m}{m!}
\int_0^r \cdots \int_0^r \bar\rho^{(m)}(t_1,\ldots,t_m)\de t_1 \cdots \de t_m
,
\\
H_{\rm inhom}^L(r)
&= H^L(r)
= -\sum_{m=1}^{\infty} \frac{(-1)^m}{m!}
\int_0^{r}
\cdots
\int_0^{r}
\frac{\bar\rho^{(m+1)}(0,t_1,\ldots,t_m)}{\rho}
\de t_1\cdots\de t_m
,
\\
J_{\rm inhom}^L(r)
&= 
J^L(r)
=
\frac{1-H^L(r)}{1-F^L(r)}.
\end{align*}
\end{cor}

It should further be noted that under homogeneity, and thereby under pseudostationarity, we have
\begin{align*}
& F^L(r) 
= 
1 - \E\left[
\prod_{x \in X} \left(
1 -
\1\{x\in b_L(u,r)\}w_{d_L}(u,d_L(u,x))
\right)
\right],
\\
&H^L(r) 
= 
1 - \E_u^!\left[
\prod_{x \in X} \left(
1 -
\1\{x\in b_L(u,r)\}w_{d_L}(u,d_L(u,x))
\right)
\right],
\end{align*}
for any arbitrary $u$, 
since $\rho(\cdot)\equiv\rho>0$. 

% \begin{exmp} 
% Any Poisson process with non-negative intensity function has $g_m(\cdot)\equiv1$, $m\geq1$, and is therefore IRMPS.
% \end{exmp}

%We have already noted how our summary statistics behave for Poisson processes. We next look closer at log-Gaussian Cox processes.

%\begin{exmp} 
%Determinantal PPs? Gibbs/Markov?
%\end{exmp}

We next look closer at how our summary statistics are affected by independent thinning. The proof of Lemma \ref{LemmaThinning} can be found in the Appendix.

\begin{lemma}\label{LemmaThinning}
Let $X$ be IRMPS and let 
% a simple point process on $L$ for which all product densities $\rho^{(m)}$, $m\geq 1$ exist and $\bar{\rho}=\inf_{u\in L}\rho(u)>0$. Let 
$X_{th}$ be an independently thinned version of $X$, generated through some measurable retention probability function $p:L \to (0,1]$.
The inhomogeneous geometrically corrected linear empty space and linear nearest neighbour distance distribution functions for $X_{th}$ are of the form
\begin{align*}
    F_{\rm inhom}^{L,th}(r)
&= -\sum_{m=1}^{\infty}
\frac{(-\bar{\rho}\bar{p})^m}{m!}
\int_0^r \cdots \int_0^r \bar g_{m}(t_1,\ldots,t_m)\de t_1 \cdots \de t_m
,
\\
H_{\rm inhom}^{L,th}(r)&= -\sum_{m=1}^{\infty} \frac{(-\bar{\rho}\bar{p})^m}{m!}
\int_0^{r}
\cdots
\int_0^{r}
\bar g_{m+1}(0,t_1,\ldots,t_m)
\de t_1\cdots\de t_m
,
\end{align*}
where %$\overline{\rho p}=\bar{\rho}\bar{p}$, $\bar{\rho}=\inf_{u\in L}\rho(u)>0$ and 
%$\bar{p}$ is the infimum of the retention probability function over $L$.
$\bar{p}=\inf_{u\in L}p(u)>0$. 
We further have that
\begin{align*}
F_{\rm inhom}^{L,th}(r) 
&= 
% 1 - \E\left[
% \prod_{x \in X} \left(1-p(x)+p(x)\left(
% 1 -
% \frac{\bar\rho\1\{x\in b_L(u,r)\}w_{d_L}(u,d_L(u,x))}{\rho(x)}
% \right)\right)
% \right]
% \\
% &=
1 - \E\left[
\prod_{x \in X_{th}} \left(
1 -
\frac{\bar{\rho}\bar{p}\1\{x\in b_L(u,r)\}w_{d_L}(u,d_L(u,x))}{p(x)\rho(x)}
\right)
\right]
\\
&=
1 - \E\left[
\prod_{x \in X} \left(
1 -
p(x)\frac{\bar\rho\1\{x\in b_L(u,r)\}w_{d_L}(u,d_L(u,x))}{\rho(x)}
\right)
\right]
\end{align*}
and
\begin{align*}
H_{\rm inhom}^{L,th}(r) 
&= 
1 - \E_u^!\left[
\prod_{x \in X_{th}} \left(
1 -
\frac{\bar{\rho}\bar{p}\1\{x\in b_L(u,r)\}w_{d_L}(u,d_L(u,x))}{p(x)\rho(x)}
\right)
\right]
\\
&=
1 - \E_u^!\left[
\prod_{x \in X} \left(
1 -
p(x)\frac{\bar\rho\1\{x\in b_L(u,r)\}w_{d_L}(u,d_L(u,x))}{\rho(x)}
\right)
\right].
\end{align*}
\end{lemma}

\subsubsection{Poisson processes}
Poisson processes serve many different purposes in spatial statistics and one of them is as benchmark for complete randomness. 

The first thing we note is that for any Poisson process with intensity $\rho(u)$, $u\in L$, $\bar\rho=\inf_{u\in L}\rho(u)>0$, we have $\xi_m(\cdot)=0$ and $g_m(\cdot)=1$, $m\geq2$, since $\rho^{(n)}(u_1,\ldots,u_n)=\rho(u_1)\cdots\rho(u_1)$. Hence, it is automatically IRMPS and the series expansions in the expressions for $F_{\rm inhom}^L(r;u)$ and $H_{\rm inhom}^L(r;u)$ are just the Taylor expansions of $1-\exp\{-\bar\rho r\}$, so $J_{\rm inhom}^L(r;u) = 1$. In particular, for a homogeneous Poisson process with constant intensity $\rho>0$, $F_{\rm inhom}^L(r;u)=H_{\rm inhom}^L(r;u)=1-\exp\{-\rho r\}$.
Note further that for a linear network $L$ which is isometric to $\R$ (i.e.~$\R$ bent in a number of places), $F_{\rm inhom}^L(r;u)=H_{\rm inhom}^L(r;u)=1-\exp\{-\rho r\} \neq 1-\exp\{-2\rho r\} = F_{\rm inhom}^{\R}(r;u) = H_{\rm inhom}^{\R}(r;u)$ for a homogeneous (and thus pseudostationary) Poisson process.

\subsubsection{Log-Gaussian Cox processes}
In the Euclidean context, log-Gaussian Cox processes \citep{moller1998lgcp} are the most prominent clustering models. We next look closer at IRMPS log-Gaussian Cox processes.

%{\color{red}
Assume that there exists a well-defined covariance function which satisfies  $C(u_1,u_2)=\mathcal{C}(d_L(u,u_1),d_L(u,u_2))\in\R$, $u_1,u_2\in L$, for any $u\in L$ and some function $\mathcal{C}$; 
%is a well defined covariance function on $L$; 
note that the dependence associated to two locations is determined by the locations' respective $d_L$-distances to some arbitrary point $u\in L$. 

% {\color{red}Let us think about what happens if $d_L(u_1,u_2)=\|u_1-u_2\|_{\R^2}$ -- is there any geometrical identity which lets us combine $\|u_1-u\|_{\R^2}$ and $\|u-u_2\|_{\R^2}$ to obtain $\|u_1-u_2\|_{\R^2}$. Note that $\mathcal{C}$ need not be a covariance function in itself (but a function of one)...}

Let $X$ be a log-Gaussian Cox process with (a.s.~locally finite) random intensity measure $\Gamma(A)=\int_A\Lambda(v)\de_1v=\int_A\exp\{Z(v)\}\de_1v$, $A\subseteq L$, where $Z$ is a Gaussian random field on $L$ with mean function $\mu(v)\in\R$, $v\in L$, and 
%isotropic 
covariance function $C(\cdot)$.
%$C(u,v)$, $u,v\in L$. %in accordance with \citet{anderes2017isotropic}. 

Note first that by e.g.~\citet[Example 5.3]{CSKWM13},
\begin{align*}
\rho^{(m)}(u_1,\ldots,u_m)
&=\E[\Lambda(u_1)\cdots \Lambda(u_m)]
=\E\left[\exp\left\{\sum_iZ(u_i)\right\}\right]
\\
&
=
\prod_i\rho(u_i)\prod_{1\leq i < j \leq m}g(u_i,u_j),
\end{align*}
whereby
\begin{align*}
g_m(u_1,\ldots,u_m)
&=
\frac{\rho^{(m)}(u_1,\ldots,u_m)}{\rho(u_1)\cdots\rho(u_m)}
=
\prod_{1\leq i < j \leq m}g(u_i,u_j)
.
\end{align*}
Now, due to the way the covariance function is defined, 
\begin{align*}
\rho(v)&=
\rho^{(1)}(v)
=
\E\left[\exp\left\{Z(v)\right\}\right]
=
\exp \left\{ \mu(v)+
\mathcal{C}(d_L(u,v),d_L(u,v))/2 \right\},
\end{align*}
and
\begin{align*}
g(u_1,u_2)
&=\exp \{ C(u_1,u_2) \}=\exp \{ \mathcal{C}(d_L(u,u_1),d_L(u,u_2))\},
\end{align*}
so 
\begin{align*}
g_m(u_1,\ldots,u_m)
&=
% \prod_{1\leq i<j \leq m}\bar g(d_L(u,u_i),d_L(u,u_j))
% \\
% &=
\prod_{1\leq i<j \leq m}\exp\left\{\mathcal{C}(d_L(u,u_i),d_L(u,u_j))\right\}
\\
&=
\exp\left\{
\sum_{1\leq i<j\leq m}\mathcal{C}(d_L(u,u_i),d_L(u,u_j))\right\},
\end{align*}
and we see that the latter is a function of the form $\bar g_m(d_L(u,u_1),\ldots,d_L(u,u_m))$, which by assumption does not change with $u$. Hence, $X$ is IRMPS.

\subsection{Non-parametric estimation}
We next turn to the non-parametric estimation of our newly defined summary statistics, based on an IRMPS point process $X$ on a linear network $L$; note that in practice $L$ is often a subnetwork of a larger network and the observed point pattern is assumed to be a realisation of $X$. In analogy with \citet{van11,Cronie2015,Cronie2016} we will focus on minus sampling estimators. %, and for ease of exposition we will assume that the intensity function is known.
% This section is devoted to propose non-parametric estimators for the \emph{pseudo empty space function} $F_{\rm inhom}^L$, the {\em pseudo nearest neighbour distance distribution function} $H_{\rm inhom}^L$ and the {\em pseudo $J$-function} $J_{\rm inhom}^L$ in the Definition \ref{SumStatsL}. 

Recalling that the boundary $\partial L$ of a linear network $L$ is the set of all nodes with degree 1, given a regular distance $d_L$, we define the $r$-erosion of $L$ (or simply the $r$-reduced network), $r\geq0$, as 
\begin{eqnarray*}
 L_{\ominus r}=
 \{u \in L: 
 d_L(u,\partial L) 
 \geq r \}, 
%  \{u \in L: b_L(u,r) \subseteq L \}, 
\end{eqnarray*}
where $d_L(u,A)=\inf_{v\in A}d_L(u,v)$, $A\subseteq L$, $u\in L$.
%
% Before turning into the non-parametric estimators, we introduce the reduced network $L_{\ominus r}$ as a set of all $u \in L$ that they are at least $r\geq0$ units far from the boundary of $L$. We recall that the boundary of $L$ is the set of all nodes with degree 1. The reduced network or $r$-erosion of $L$, $L_{\ominus r}$ is then of the form
% \begin{eqnarray*}
%  L_{\ominus r}=\{u \in L: b_L(u,r) \subseteq L \}, \qquad r>0.
% \end{eqnarray*}
%
We further let $I$ be a set consisting of a large number of fixed points in $L$; 
% Further let $I$ be a set of fixed points in $L$; 
this is analogous to a point grid in $\R^d$. 

For a given $r\geq0$, we estimate $F_{\rm inhom}^L(r)$, i.e.~the inhomogeneous geometrically corrected linear empty space function evaluated at $r$, by means of
% For an observed point process $X$, the estimator of $F_{inhom}^L$ is of the form
\begin{align}\label{eq:estFinhom}
\widehat F_{\rm inhom}^L(r)
=
1 - \frac{1}{N(I \cap L_{\ominus r})}
\sum_{u \in I \cap L_{\ominus r} } 
\prod_{x \in X \cap b_L(u,r)} 
\left( 1- \frac{\bar{\rho}}{\rho(x)} w_{d_L}(u,d_L(u,x))\right),
\end{align}
and $H_{\rm inhom}^L(r)$, i.e.~the inhomogeneous geometrically corrected linear nearest neighbour distance distribution function at $r$, by means of 
\begin{align}\label{eq:estHinhom}
\widehat H_{\rm inhom}^L(r)
=
1-
\frac{1}{N(X \cap L_{\ominus r})} 
\sum_{u \in X \cap L_{\ominus r} } 
\prod_{x \in X\setminus\{u \} \cap b_L(u,r)} 
\left( 1- \frac{\bar{\rho}}{\rho(x)} w_{d_L}(u,d_L(u,x))\right).
\end{align}
Having the estimators above at hand, we then estimate $J_{\rm inhom}^L(r)$, i.e.~the inhomogeneous geometrically corrected linear $J$-function at $r\geq0$, by means of 
\begin{align}
\label{eq:estJinhom}
\widehat J_{\rm inhom}^L(r)
=
\frac{1-\widehat H_{\rm inhom}^L(r)}
{1-\widehat F_{\rm inhom}^L(r)}
, 
\end{align}
provided that the denominator is non-zero. 
When we are working under the assumption of a pseudostationary process we alter the estimation by removing the ratios $\bar{\rho}/\rho(x)$ in \eqref{eq:estFinhom} and \eqref{eq:estHinhom} since $\rho(\cdot)\equiv\rho>0$.

Part of the motivation for considering a minus sampling estimation scheme here is that it yields (ratio) unbiasedness. The proof of Theorem \ref{ThmUnbiased} can be found in the Appendix.
\begin{thm}\label{ThmUnbiased}
The estimator \eqref{eq:estFinhom} is unbiased and \eqref{eq:estHinhom} is ratio-unbiased in the sense that both its numerator and denominator are unbiased.
\end{thm}

\subsubsection{Intensity estimation}\label{sec:IntensityEstimation}
In practice, the true intensity function is unknown so in order to exploit the estimators \eqref{eq:estFinhom} and \eqref{eq:estHinhom} we need to estimate the intensity function $\rho(\cdot)$ in advance and then plug this estimate into \eqref{eq:estFinhom} and \eqref{eq:estHinhom}. Obtaining good estimates for intensity functions of point processes on linear networks has been a challenging task due to geometrical complexities and unique methodological problems. Nevertheless, there have been a few particularly interesting proposals, including diffusion based kernel estimation \citep{mcswiggan2017}, an edge-corrected classical kernel-based intensity estimator \citep{MFJ18}, resample-smoothed Voronoi estimation \citep{Moradi2019} and fast kernel smoothing using two-dimensional convolutions \citep{rakshit2019fast}. 
Although \citet{Moradi2019} showed that their approach in general generates better intensity estimates than kernel-based approaches, we have observed that \eqref{eq:estJinhom} generally performs better numerically when it is combined with a kernel estimator. Regarding the associated bandwidth selection, one would expect that the estimator of \citet{cronie2018bandwidth} or Poisson likelihood cross-validation \citep{BRT15} would be the best choice. It seems, however, that Scott's rule of thumb \citep{scott2015multivariate,rakshit2019fast} in general yields more stable outputs of \eqref{eq:estJinhom}. 
In our numerical evaluations we make use of the fast kernel estimator of \citet{rakshit2019fast} which we combine with Scott's rule of thumb.

\section{Numerical evaluation}\label{sec:numerical}
We next numerically evaluate the performance of the estimator of $J_{\rm inhom}^L(r)$. 
%in practice, based on simulated realisations. 
For this purpose, we simulate point patterns from three different models with different types of spatial interaction -- spatial randomness, regularity and clustering. For each model we make use of two linear networks: the network of a Chicago crime dataset and the network of a dataset on spider locations, which are both publicly available in the \textsf{R} package \textsf{spatstat} \citep{baddeley2004spatstat,BRT15}.

%{\color{red}
%IS THERE AN ACTUAL POINT OF INCLUDING THIS (CAN'T IT BE FOUND ELSEWHERE/IS IT ENOUGH TO GIVE $W$)?

Here we provide some general information about both of the networks:
\begin{itemize}
    \item The Chicago linear network has $503$ segments and $338$ nodes, where $44$ of them have degree $1$, thus forming the boundary of the network. The total length of the network is $31150.21$ feet and its maximum node degree is 5. The network is embedded in the window $W=[0.3894, 1281.9863] \times [153.1035, 1276.5602]$.
    
    \item The linear network for the data on spider locations is constructed by 156 nodes and 203 segments. It has a total length of 20218.75 millimetres and a maximum node degree of 3. There are 31 nodes with degree 1. The network is embedded in the window $W=[0, 1125]^2$.
\end{itemize}{}

%}

% \begin{comment}
% {\color{red} 
% %Here we give what is needed to provide eq. (8) in \citet{rakshit2017second}, including the def. of a regular metric and the integral formula.

% Following \citet{ABN12}, given some fixed $v\in L$ we have that for any measurable function $f: L^m \to \R$, %we have that %$\int_L \cdots \int_L f(u_1,\ldots,u_m) \de_1 u_1 \cdots \de_1 u_m $ is equal to

% \begin{align*}
% &\int_L \cdots \int_L f(u_1,\ldots,u_m) \de_1 u_1 \cdots \de_1 u_m =
% \\
% &
% =
% \int_0^{\infty} \sum_{u_1 \in L:d_L(u_1,v)=t_1}
% \cdots
% \int_0^{\infty} \sum_{u_m \in L:d_L(u_m,v)=t_m}
% f(u_1,\ldots,u_m) \de t_1 \cdots \de t_m,
% \end{align*}
% where in the particular case that $f$ only depends on $d_L(u_i,v)$, $i=1,\ldots,m$, the right hand side of the above equation can be expressed as
% \begin{eqnarray}\label{eq:inttransform}
% %\int_L \cdots \int_L f(u_1,\ldots,u_m) \de_1 u_1 \cdots \de_1 u_m = 
% \int_0^{\infty} \cdots \int_0^{\infty} \left[\prod_{i=1}^m c_L(v,t_i) \right] f(t_1,\ldots,t_m) \de t_1 \cdots \de t_m.
% \end{eqnarray}
% }
% \end{comment}

For each model, and each network, we first generate one realisation and then estimate the intensity in accordance with the recommendations in Section \ref{sec:IntensityEstimation}. 
%In our numerical evaluations we make use of the fast kernel estimate to speed up the calculations.
We next estimate $J_{\rm inhom}^{L}$ based on that realisation, together with pointwise critical envelopes which are computed based on $99$ realisations of a Poisson process with the obtained intensity estimate as intensity function. 
 We do so to get an indication of how well our $J$-function estimator can reveal deviations from a Poisson process behaviour and what type of interaction the underlying model possesses. 
% {\color{blue} We have performed these steps for a large number of realisations but we have chosen to only present the output for one realisation per model. }
Throughout, following most of the previous literature on analysis of spatial interaction on linear networks, we let $d_L$ be given by the shortest-path distance.

It is important to note that in neither the log-Gaussian Cox process example nor the thinned simple sequential inhibition example below we have actually verified that the models are indeed IRMPS under $d_L$. They are models which, based one our general understanding, should exhibit clustering and inhibition, and we here want to see if our $J$-function estimator manages to capture these expected behaviours. 

\subsection{Poisson process}
We here consider an inhomogeneous Poisson process $X$ with intensity function $\rho(x,y)=0.005 |\sin(x)|$; the expected number of points on the Chicago network is $101.9$, and on the spider location network it is $62.3$. A single realisation on each network is shown in the top row of Figure \ref{fig:poiss}. The bottom row of Figure \ref{fig:poiss} shows the inhomogeneous linear $J$-function estimates for each realisation together with the corresponding pointwise critical envelopes based on $99$ simulations from a Poisson process with the estimated intensity as intensity function. 
For each of the networks it can be seen that the $J$-function estimate of the simulated pattern stays around the mean of the $J$-function estimates for the envelope processes, and it entirely remains within the envelope.

\begin{figure}[!ht]
    \centering
   \includegraphics[scale=.3]{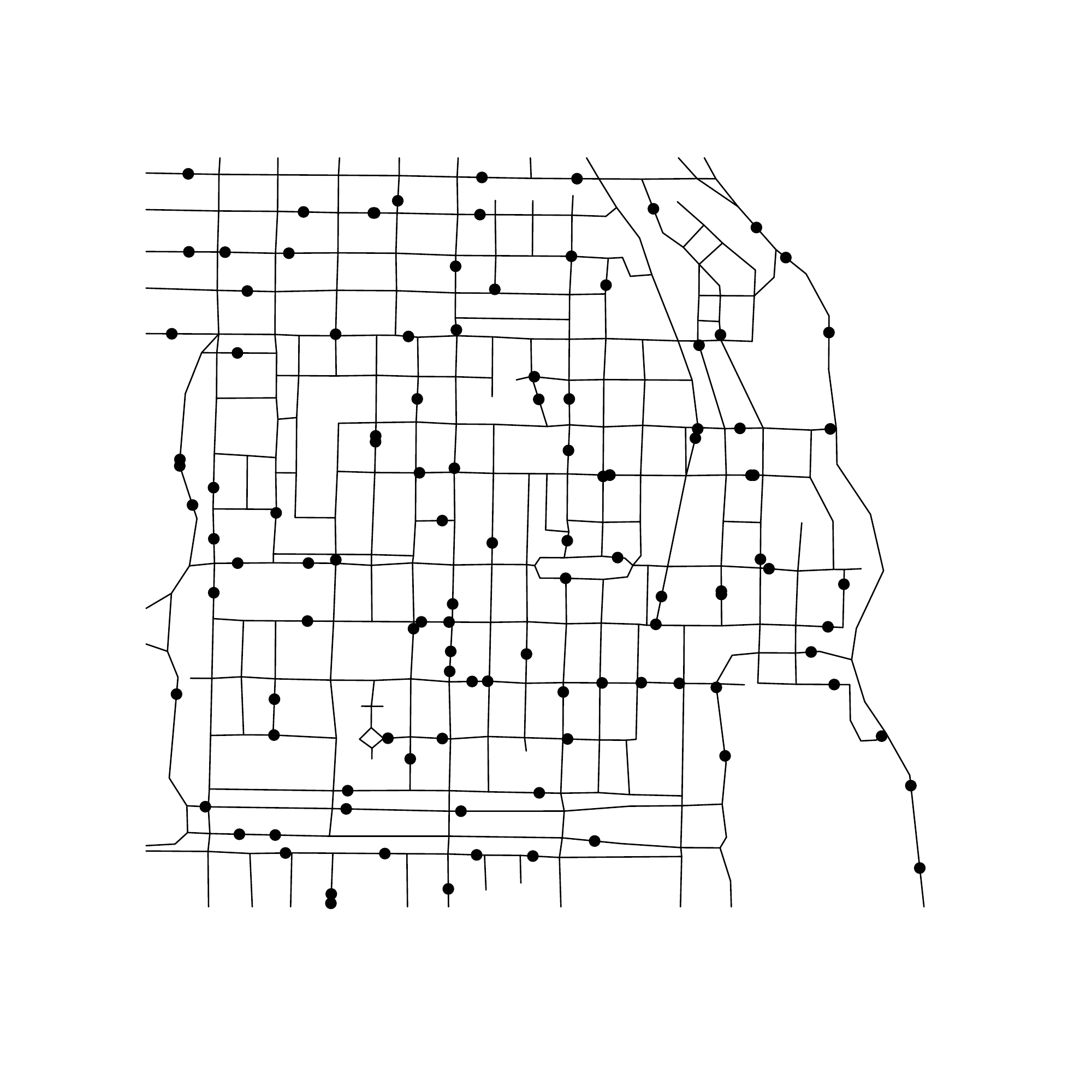}
   \includegraphics[scale=.3]{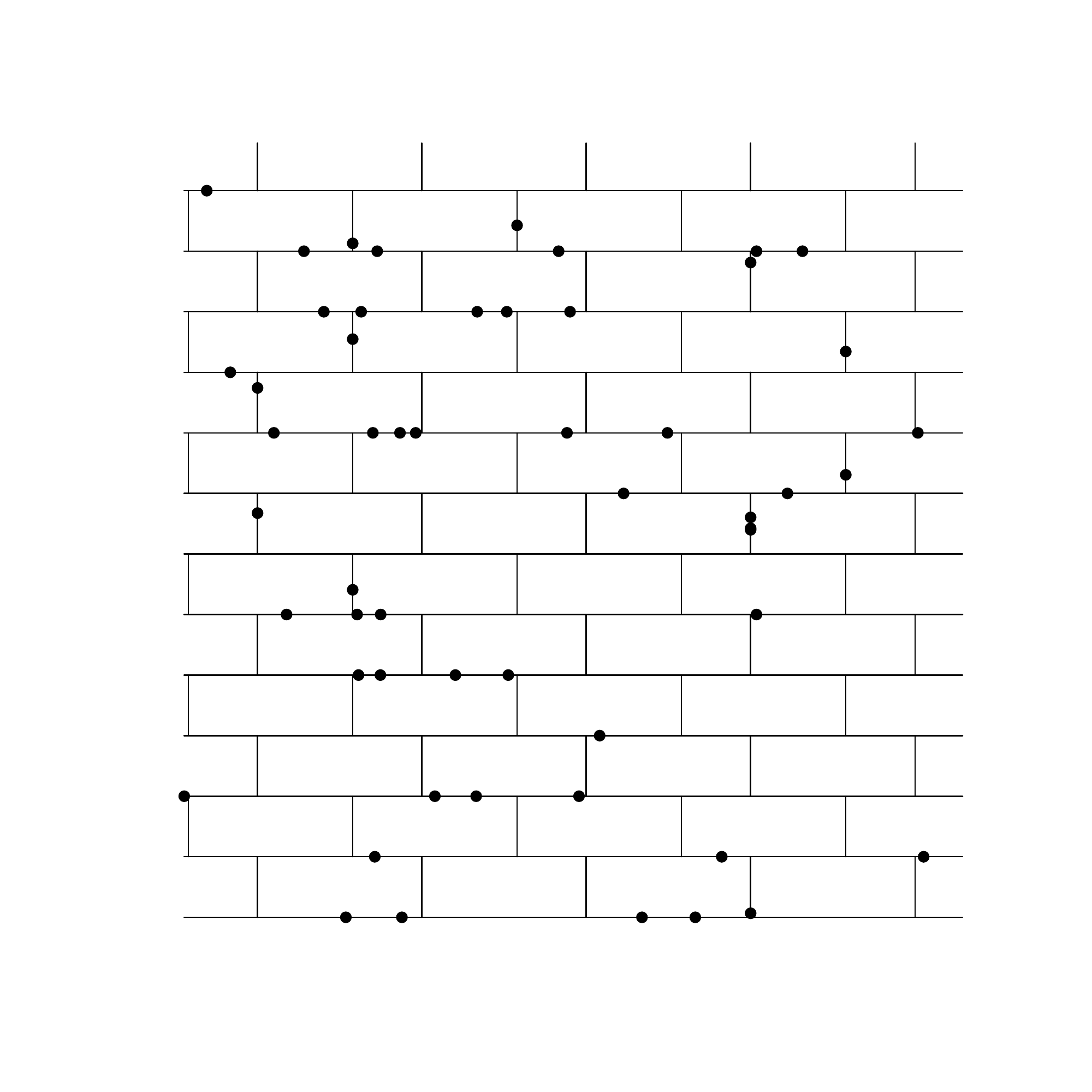}
   \includegraphics[scale=.3]{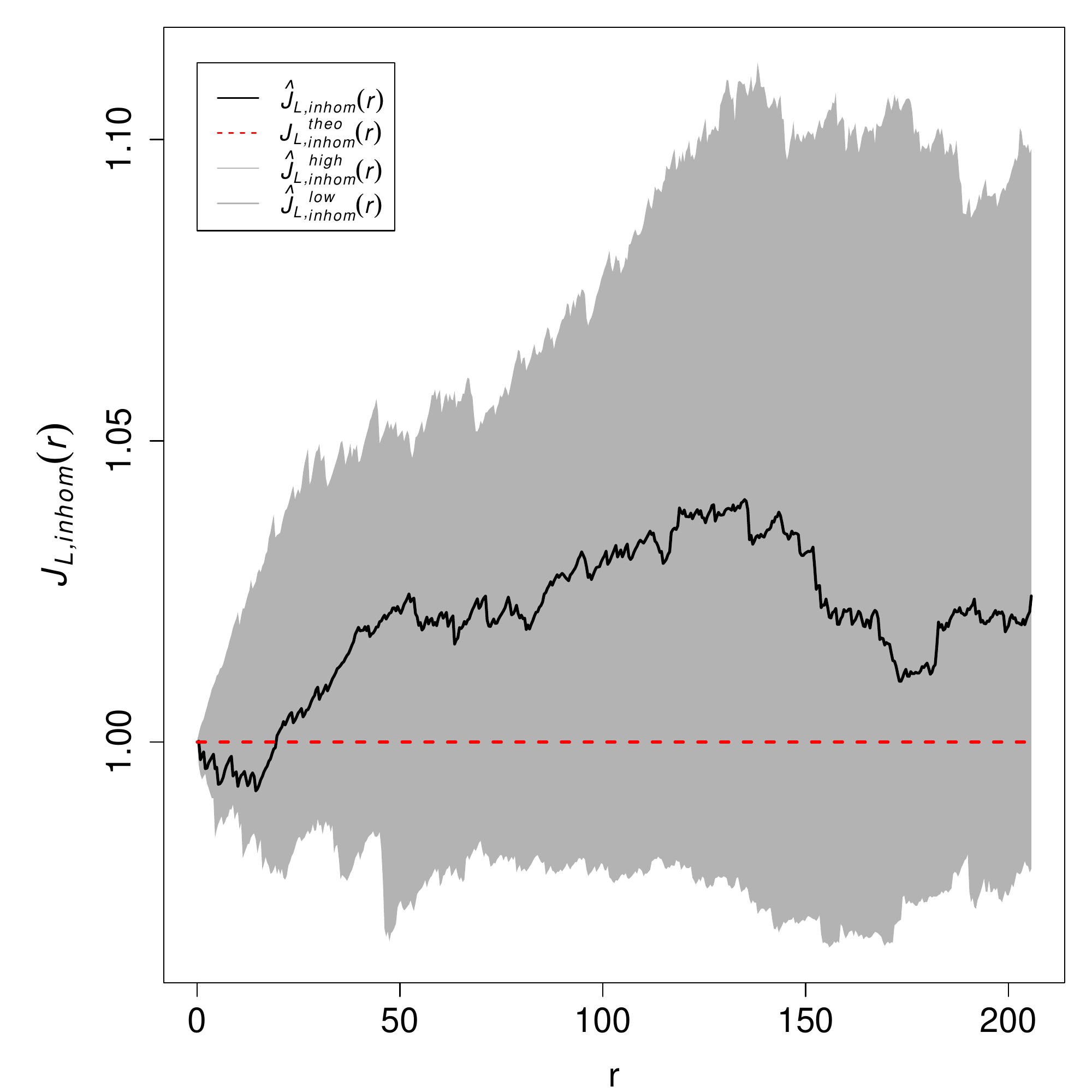}
   \includegraphics[scale=.3]{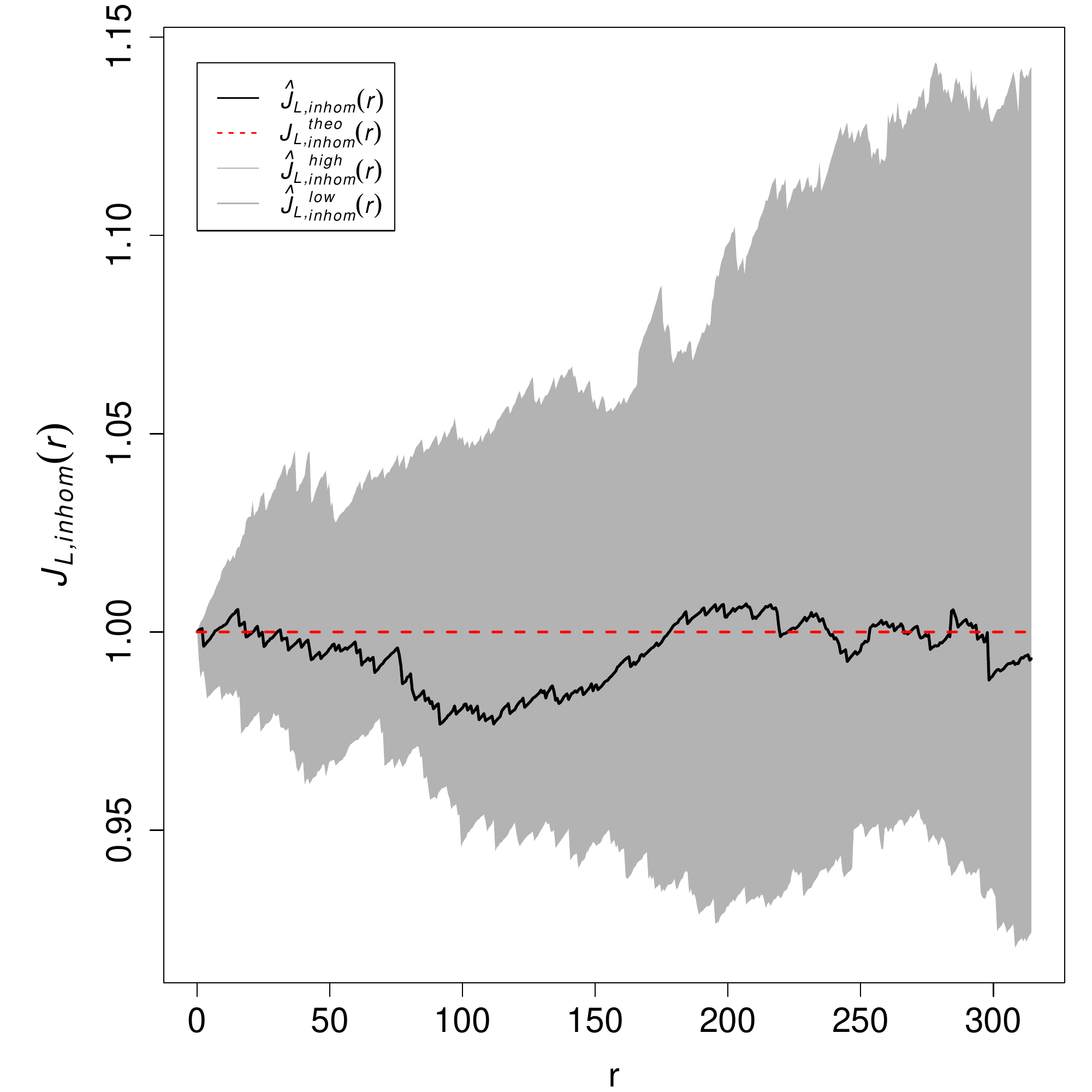}
    \caption{{\em Top row}: Realisations of inhomogeneous Poisson processes with intensity function $\rho(x,y)=0.005 |sin(x)|$ on the Chicago network (left) and on the spiders network (right). {\em Bottom row}: The corresponding inhomogeneous linear $J$-function estimates for each realisation, together with pointwise critical envelopes (grey area) based on $99$ simulations of inhomogeneous Poisson processes with the estimated intensities of the realisations in the top row as intensities. The solid lines are the estimated $J$-functions for the observed patterns and the dashed lines represent the theoretical linear $J$-function value for Poisson processes. 
    % are the $J$-function estimate means for the corresponding envelope Poisson processes. 
    Each $J$-function estimate plot is shown below its corresponding realisation.
    }
    \label{fig:poiss}
\end{figure}

\subsection{Thinned simple sequential point process}
We now consider a scenario where there is inhibition between the points. Initially we generate a realisation of a simple sequential inhibition point process with a total point count of 300 and inhibition distance $0.001|L|$; this results in an inhibition distance of 46 feet for the Chicago network, and 30 mm for the spider location network. We then thin each pattern based on the constant retention probability $p(x,y)=0.3$, $(x,y)\in L$. 
%constant retention probability 0.2. 
This results in a thinned simple sequential process with intensity function 
$\rho(x,y)=0.3(300/|L|)$, 
%$\rho(x,y)=0.2(1000/|L|)$ 
with a total expected number of points of $90$; 
%$200$; 
$|L|$ is the total length of the network. The top row of Figure \ref{fig:TSSI} shows two realisations of this process: on the Chicago network in the left panel and on the spider location network in the right panel. The corresponding estimated inhomogeneous linear $J$-functions for each realisation is displayed on the bottom row of Figure \ref{fig:TSSI}. A critical envelope (grey area) based on $99$ simulations from a Poisson process with the estimated intensity as intensity function 
%of a uniform Poisson process 
is displayed together with each estimated  $J$-functions. We see that it properly identifies an inhibitory behaviour between the points. 

It should be noted that the model here in fact is homogeneous so it may be argued that we should instead use the homogeneous estimator where we set $\bar\rho/\rho(x)=1$ in \eqref{eq:estFinhom}, \eqref{eq:estHinhom} and \eqref{eq:estJinhom}. However, since we in practice do not actually know whether a point pattern comes from a process which is homogeneous or not, we here want to see how well \eqref{eq:estJinhom} captures spatial interaction under the (incorrect) assumption that the underlying process is inhomogeneous.

\begin{figure}[!ht]
    \centering
    \includegraphics[scale=.3]{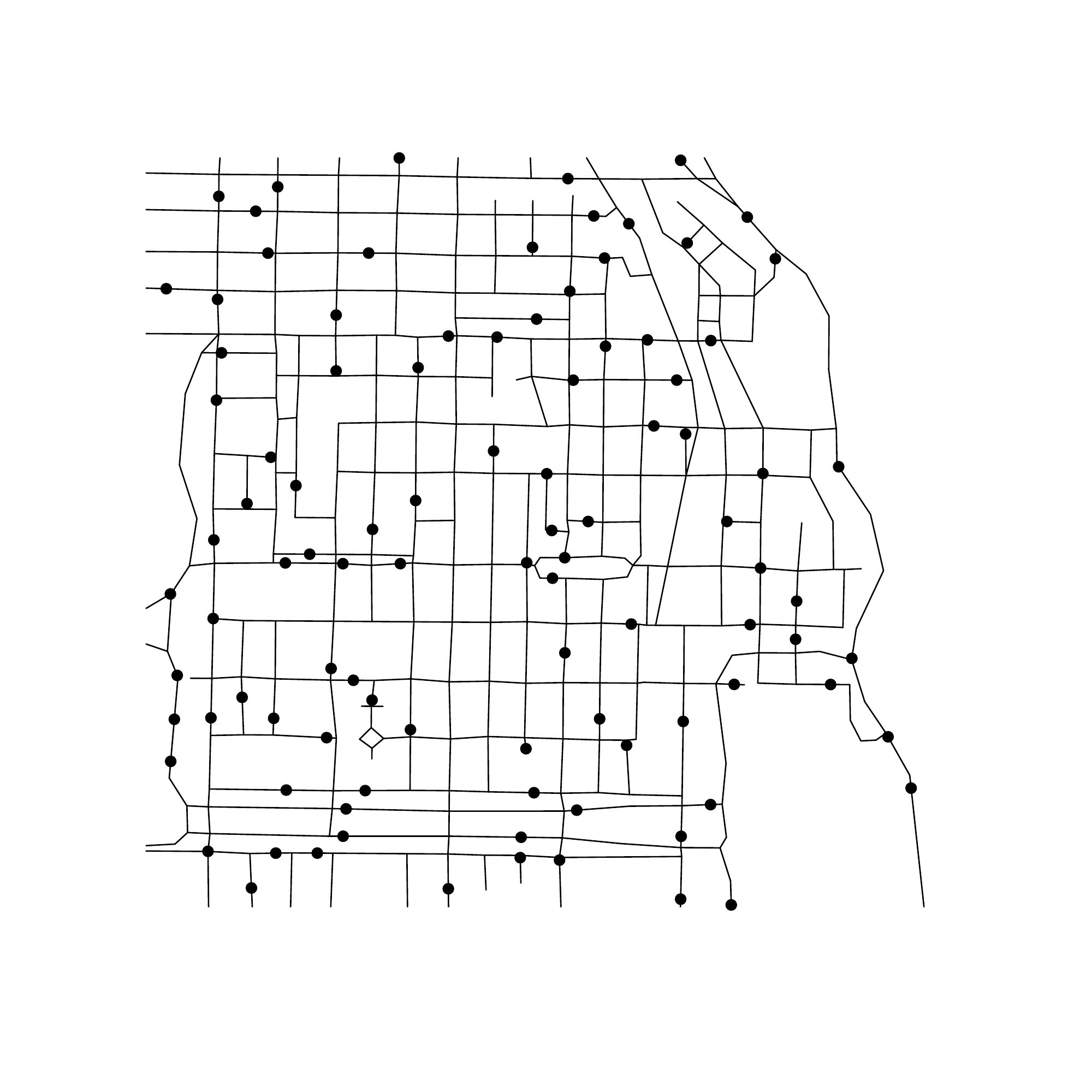}
    \includegraphics[scale=.3]{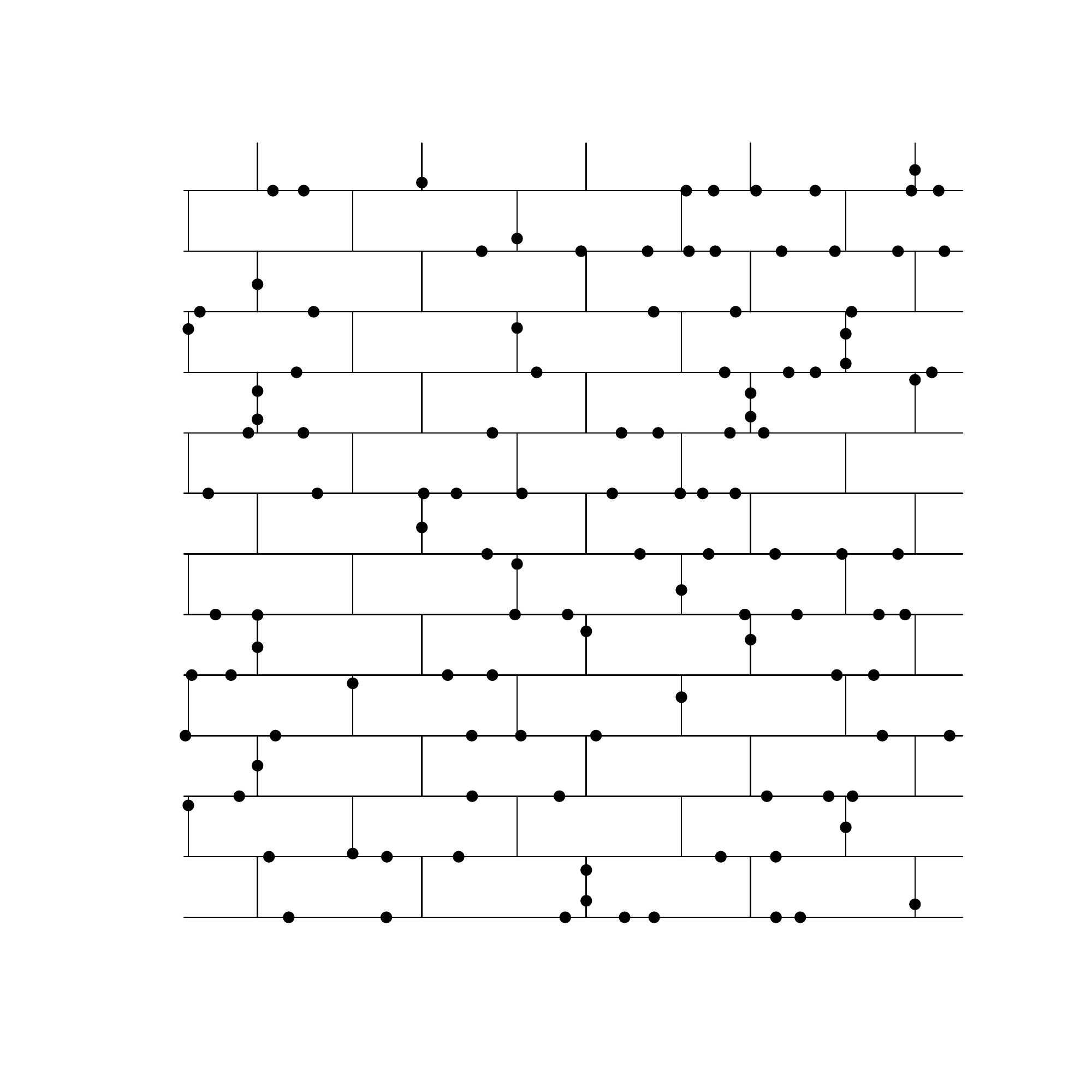}
    \includegraphics[scale=.3]{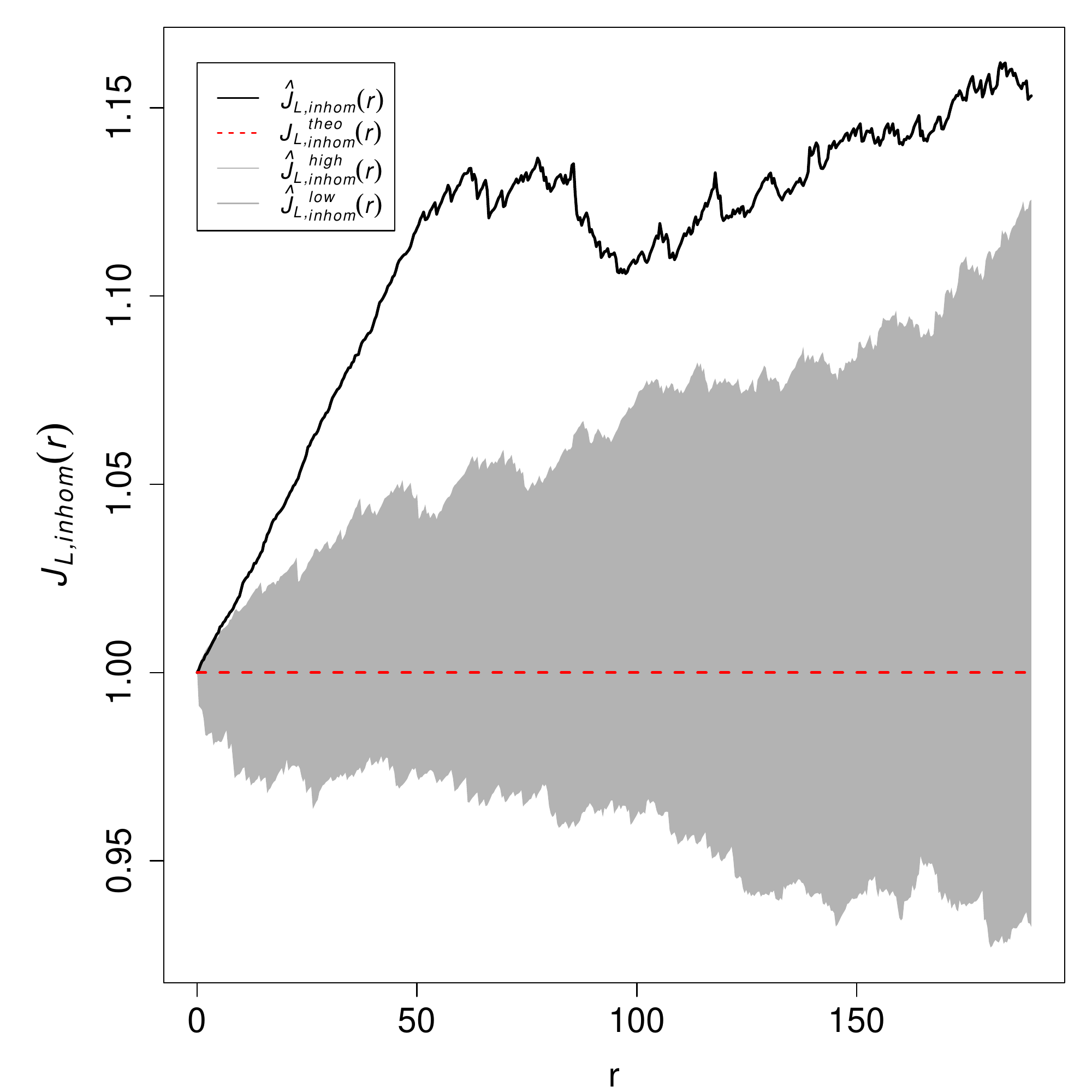}
    \includegraphics[scale=.3]{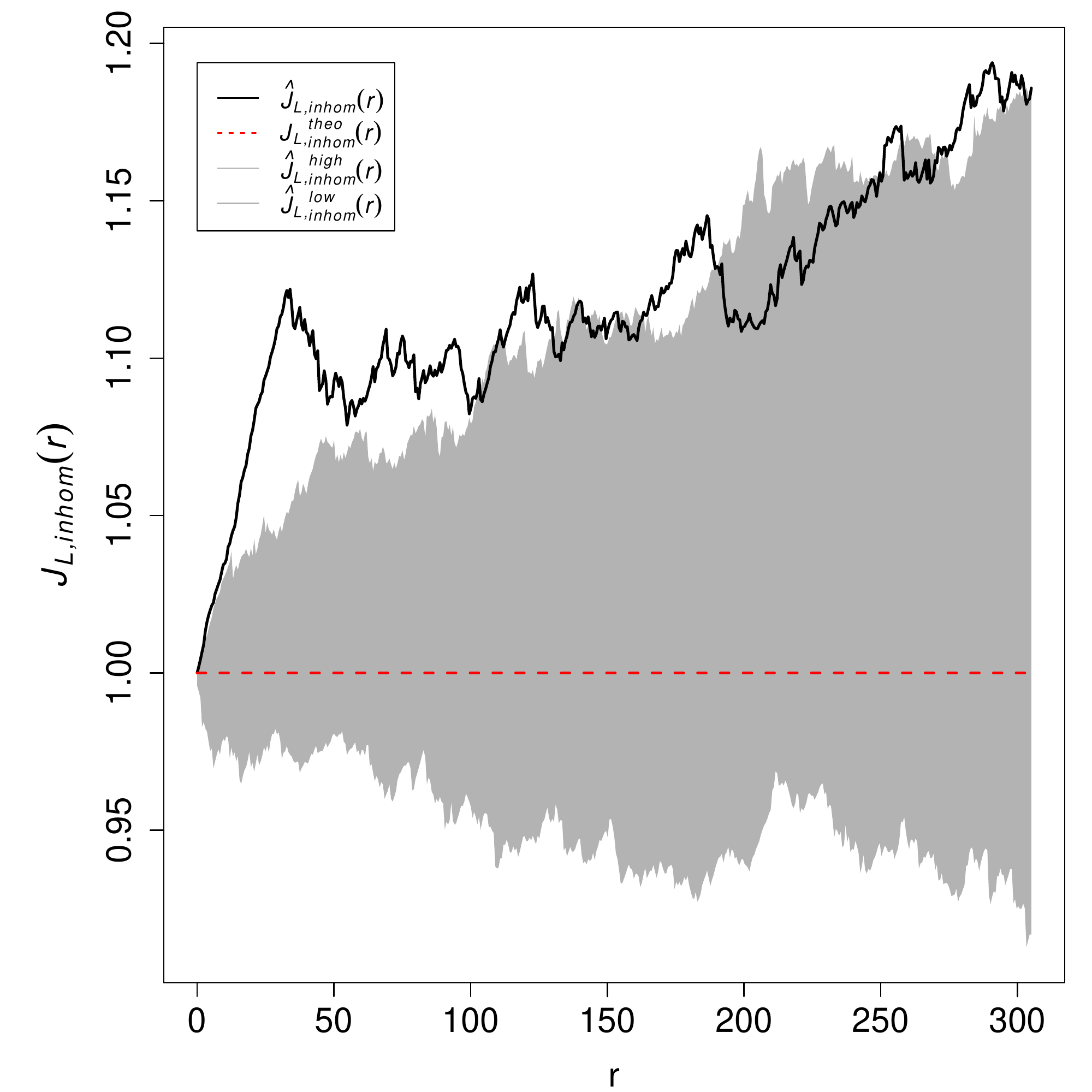}
    \caption{{\em Top row}: Realisations of  independently thinned simple sequential inhibition processes with the intensity function $\rho(x,y)=0.3(300/|L|)$; the unthinned processes have inhibition distance $0.001|L|$; 
    this results in an inhibition distance of 46 feet for the Chicago network (left) and 30 mm for the spider location network (right).
    %inhibition distance $46$ feet ({left}) and $30$ mm ({right}). {\em Bottom row}: 
    The corresponding inhomogeneous linear $J$-function estimates for each realisation, together with pointwise critical envelopes (grey area) based on $99$ simulations of inhomogeneous Poisson processes with the estimated intensities of the realisations in the top row as intensities. The solid lines are the estimated $J$-functions for the observed patterns and the dashed 
    lines represent the theoretical linear $J$-function value for Poisson processes. 
    % lines are the $J$-function estimate means for the corresponding envelope Poisson processes. 
    Each $J$-function estimate plot is shown below its corresponding realisation.
    % The corresponding inhomogeneous $J$-functions for each realisation together with pointwise critical envelopes (grey area) based on $99$ simulations of a uniform Poisson process. The solid line is the estimated $J$-function for the observed pattern and the dashed line is the $J$-function for Poisson processes. Each $J$-function is shown below its corresponding realisation.
    }
    \label{fig:TSSI}
\end{figure}

\subsection{Log-Gaussian Cox process}
In this section we first generate a realisation of a log-Gaussian random field on the window $W=[x_{\min},x_{\max}]\times[y_{\min},y_{\max}]$ and then evaluate it only at locations on the network $L\subseteq W$ in question. We then use this extracted realisation to simulate a realisation of an inhomogeneous Poisson process on the network. The driving Gaussian random field on $W$ has
mean function $(x_1,y_1)\mapsto \log 0.002 + (x_1-(x_{\max}-x_{\min}))/|L|$ and covariance function $ ((x_1,y_1),(x_2,y_2))\mapsto \exp \{- \| (x_1,y_1)-(x_2,y_2) \| \}$, $(x_1,y_1),(x_2,y_2)\in W$. Hereby, the intensity is given by $\rho(x,y) = 0.002 \exp\{ [(x_1-(x_{\max}-x_{\min}))/|L|] + 0.5 \}$, $(x,y)\in L$. The expected number of points on the Chicago network is 101.1, and for the spiders network it is 64.2. The top row of Figure \ref{fig:lgcp} shows two realisations of such a model on the Chicago network (left) and the spiders network (right). The corresponding estimated inhomogeneous linear $J$-function for each realisation is exhibited on the bottom row of Figure \ref{fig:lgcp}. A critical envelope (grey area) based on 99 realisations from a Poisson process with the estimated intensity as intensity functions is displayed together with each estimated $J$-function. 
From Figure \ref{fig:lgcp} we can see that the $J$-function estimate stays below the envelope for small and moderate interaction ranges $r$, thus indicating a clustering behaviour for the underlying model. 

\begin{figure}[!ht]
    \centering
   \includegraphics[scale=.3]{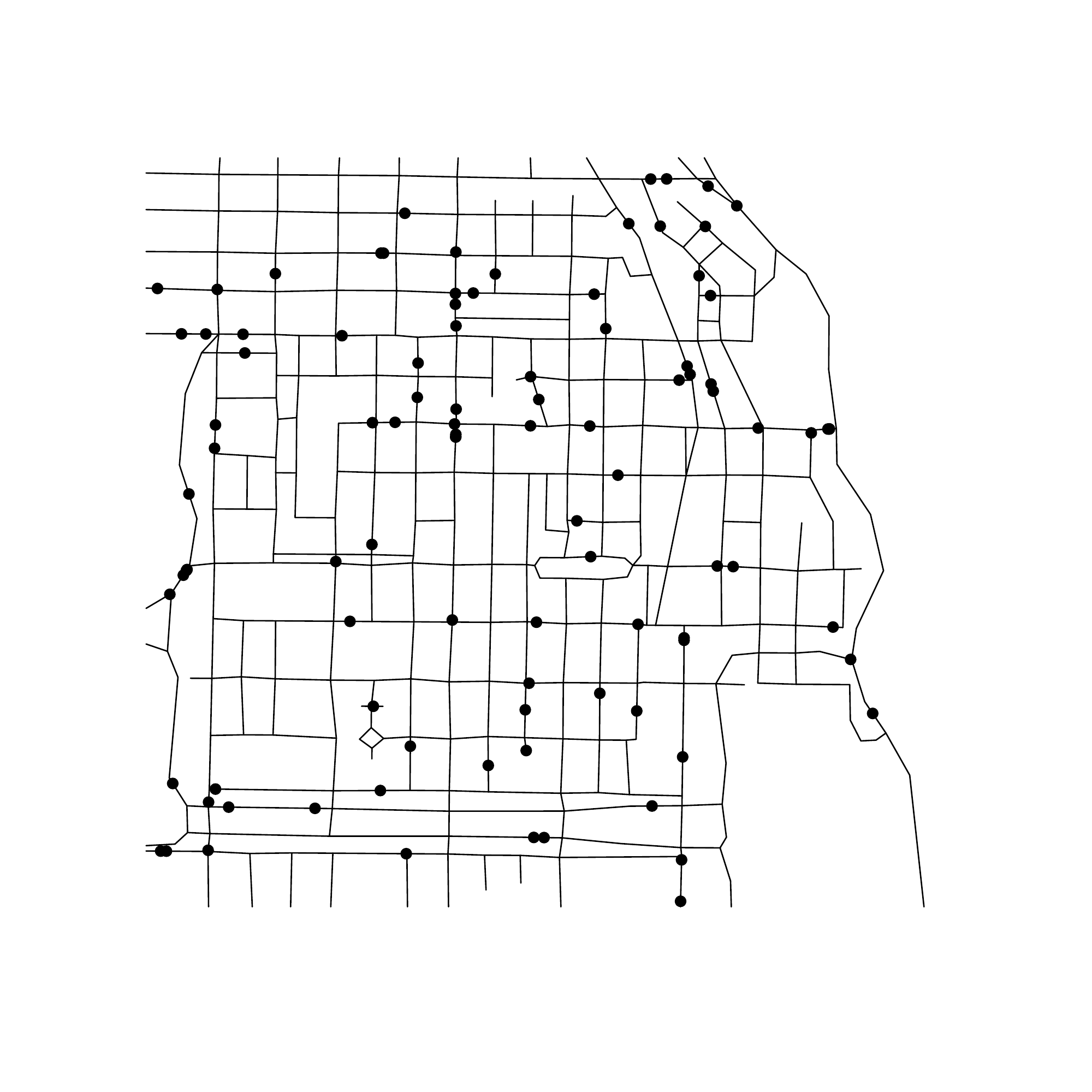}
   \includegraphics[scale=.3]{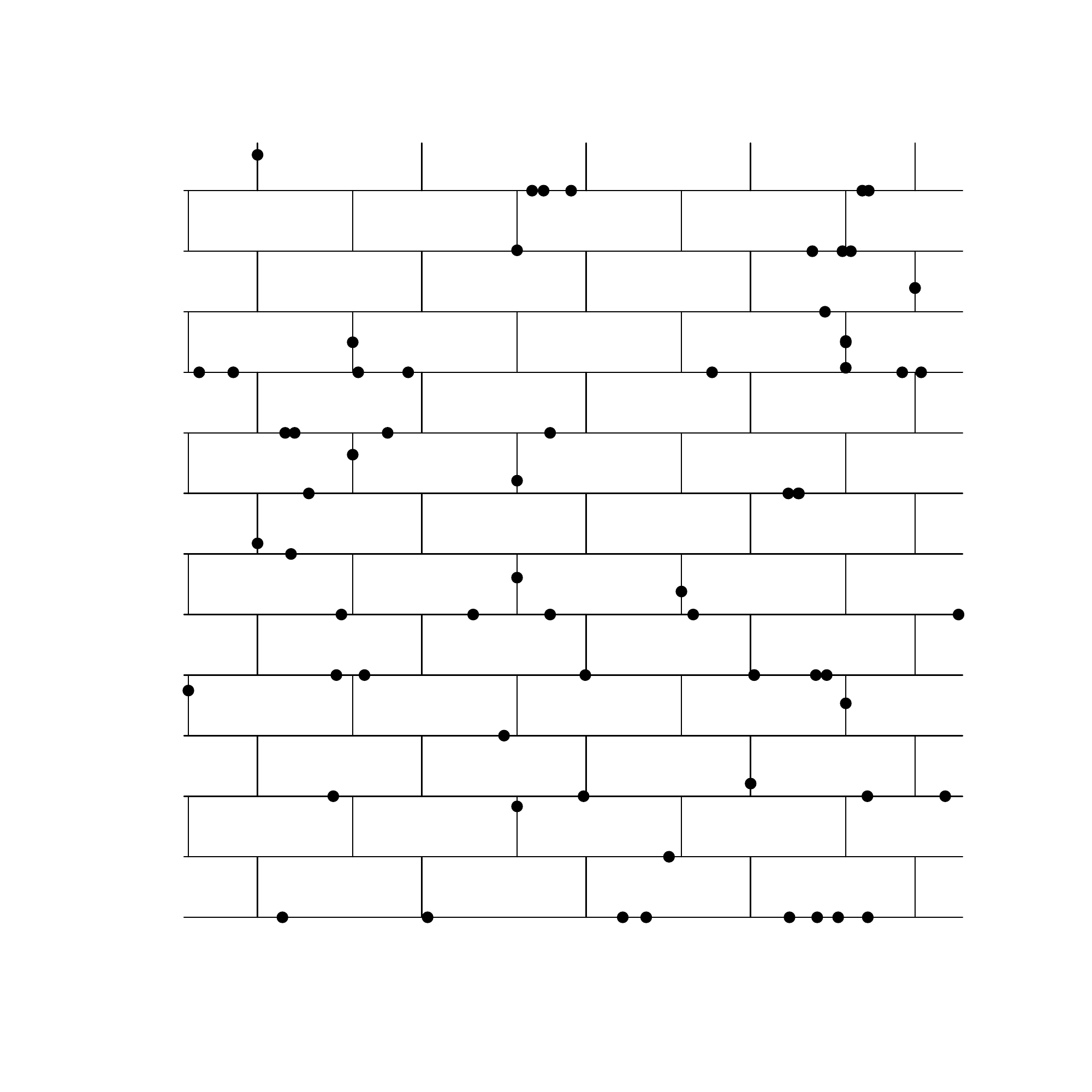}
   \includegraphics[scale=.3]{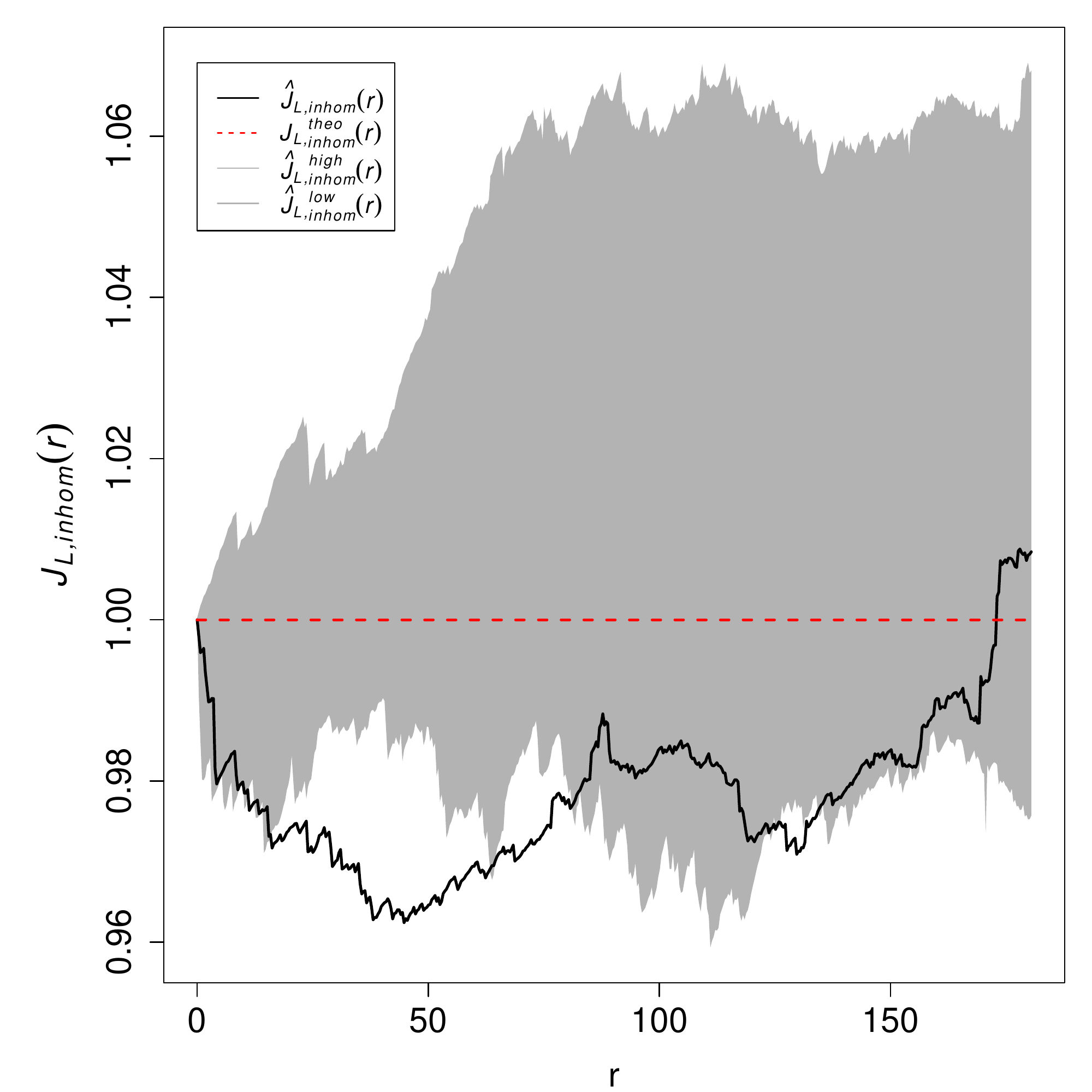}
   \includegraphics[scale=.3]{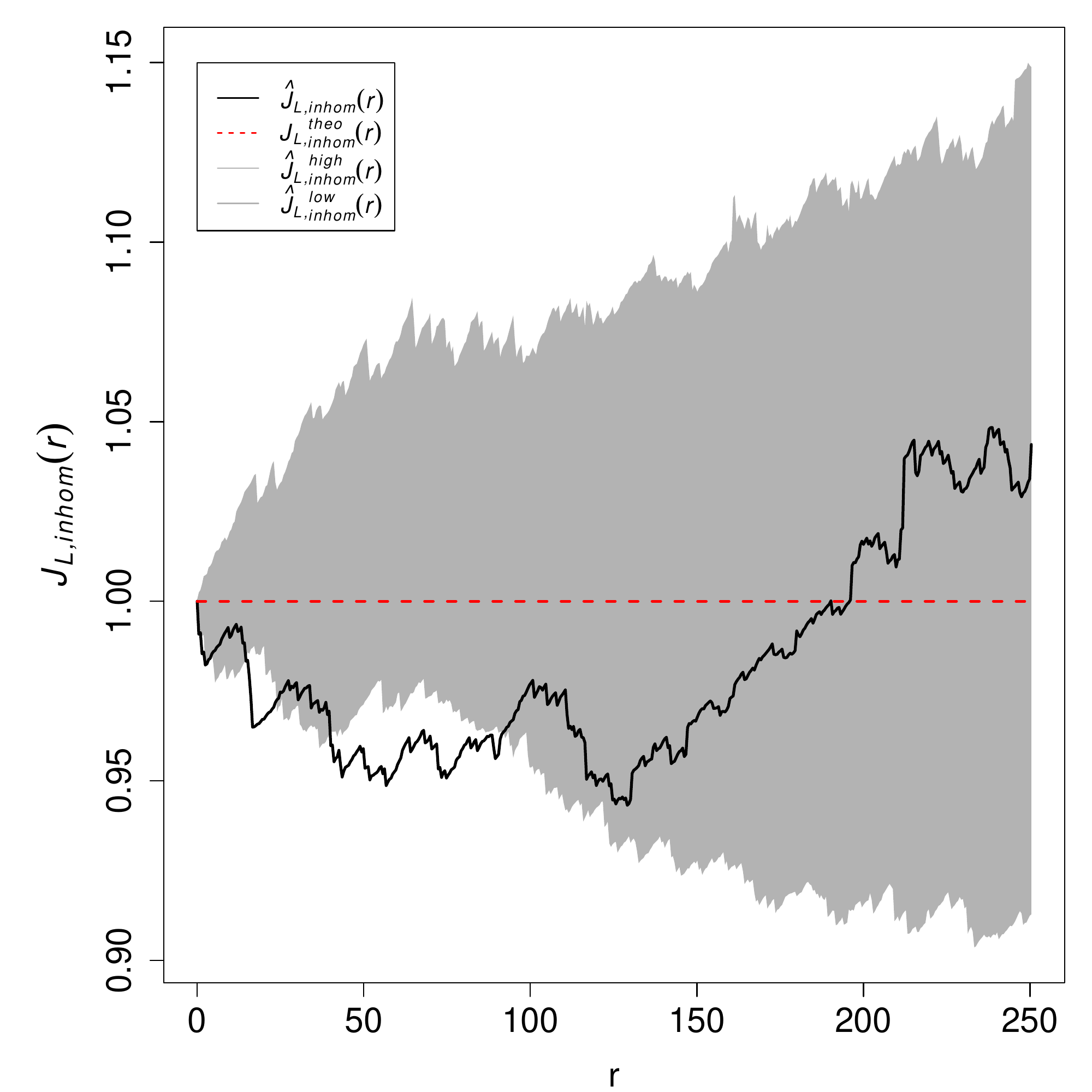} \caption{{\em Top row}: Realisations of log Gaussian Cox process models on the Chicago network (left) and on the spiders network (right). {\em Bottom row}: The corresponding inhomogeneous linear $J$-functions for each realisation together with pointwise critical envelopes (grey area) based on $99$ simulations of inhomogeneous Poisson processes with the estimated intensities of the realisations in the top row as intensities. The solid lines are the estimated $J$-functions for the observed patterns and the dashed lines represent the theoretical linear $J$-function value for Poisson processes. Each $J$-function plot is displayed below its corresponding realisation.
    }
    \label{fig:lgcp}
\end{figure}

%\clearpage
\section{Data analysis}\label{sec:data}
We next apply the inhomogeneous linear $J$-function estimator to the two real datasets in Figure \ref{fig:data}: a) a point pattern of motor vehicle traffic accidents in an area of Houston, US, which was previously studied in \cite{Moradi2019}, b) the spider dataset which represents the locations of webs made by spiders in the mortar spaces of a brick wall -- this dataset has also previously been studied in \cite{ABN12}. 
As in the case of the numerical evaluations, we here let $d_L$ be given by the shortest-path distance.

\subsection{Houston motor vehicle traffic accidents}
%{\color{red}
The right panel in Figure \ref{fig:data} shows the locations of $249$ traffic accidents in an area of Houston, US, during April, $1999$. The linear network $L$ has a total length of $708301.7$ feet, and has $187$ nodes with a maximum node degree of $4$, and $253$ line segments. 
For further details, see \citet{levine2006, levine2009} and \citet{Moradi2019}.
%The data is collected by individual police departments in the Houston metropolitan area and compiled by the Texas Department of Public Safety. The compiled data have been obtained by the Houston-Galveston Area Council and then geocoded by N.\ Levine.
%}
\citet{Moradi2019} studied intensity estimation on this dataset, using their resample-smoothed Voronoi intensity estimator. Figure \ref{fig:Houston} shows that the estimated inhomogeneous linear $J$-function is almost entirely inside the pointwise critical envelope area which has been computed based on $99$ simulations of a Poisson process with intensity given by the estimated intensity (which is obtained in analogy with the numerical evaluation section). 
Since the estimated $J$-function stays within the envelopes (except at the very end) there seem to be no clear indications of clustering/inhibition.

% Apparently, it does not identify any deviation from being Poisson, especially for small distances.

\subsection{Spiders data}\label{sec:spider}
The left panel in Figure \ref{fig:data} shows the locations of 48 webs of the urban wall spider Oecobius navus on the mortar lines of a brick wall. This dataset was recorded by \cite{voss1999habitat} and it is stored in the \textsf{R} package \textsf{spatstat} \citep{baddeley2004spatstat,BRT15}. It has previously been studied by \cite{ABN12} through second-order summary statistics. The right panel of Figure \ref{fig:Houston} shows the estimated inhomogeneous linear $J$-function for this dataset together with a pointwise critical envelope based on $99$ simulations of a Poisson process with the estimated intensity as intensity function (which is obtained in analogy with the numerical evaluation section). The estimated $J$-function is fully inside the envelope and does not indicate any deviations from being Poisson -- this is in keeping with \citet{ABN12}.

\begin{figure}[!ht]
    \centering
    % \includegraphics[scale=.32]{Houston.pdf}
    % \qquad
    \includegraphics[scale=.3]{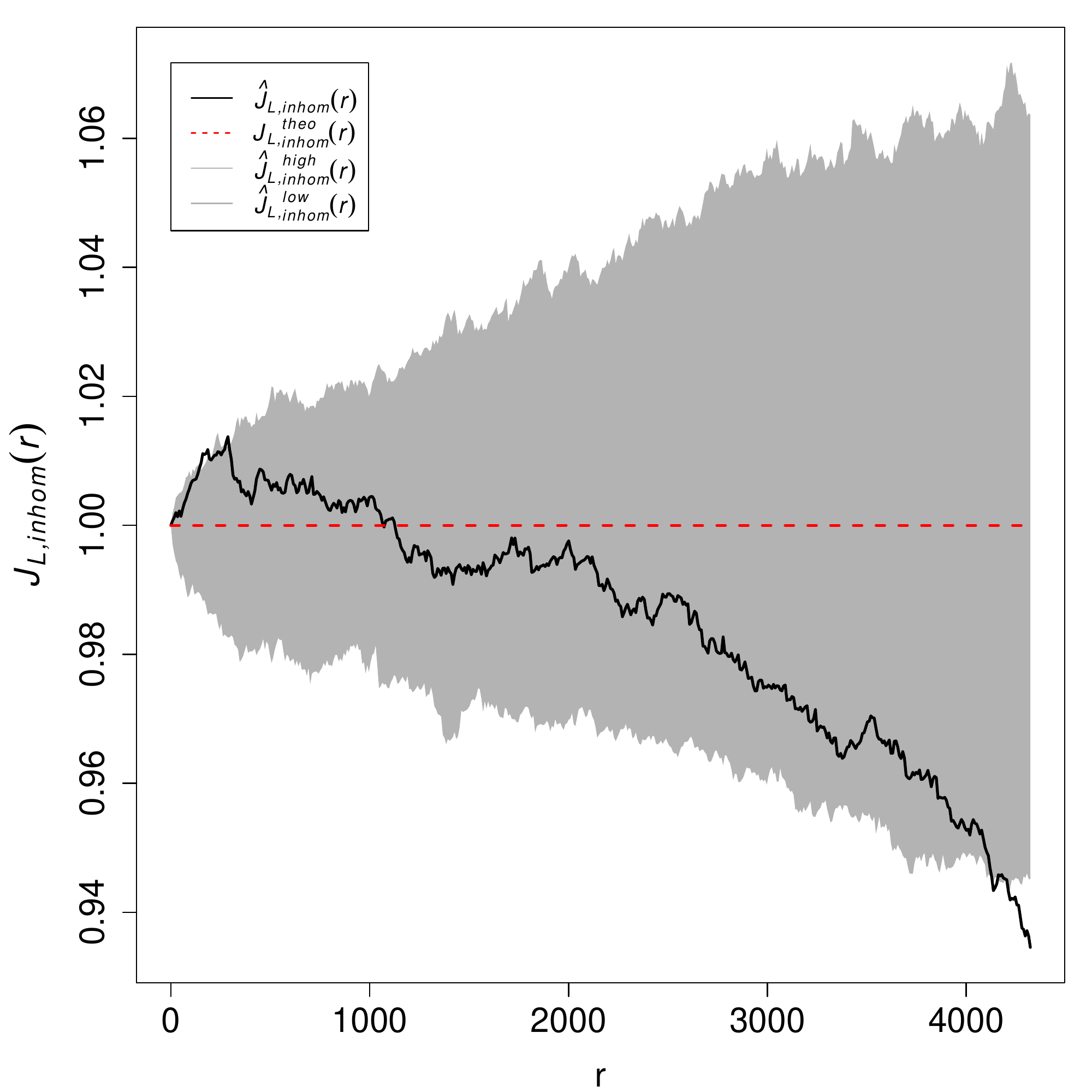}
    \includegraphics[scale=.3]{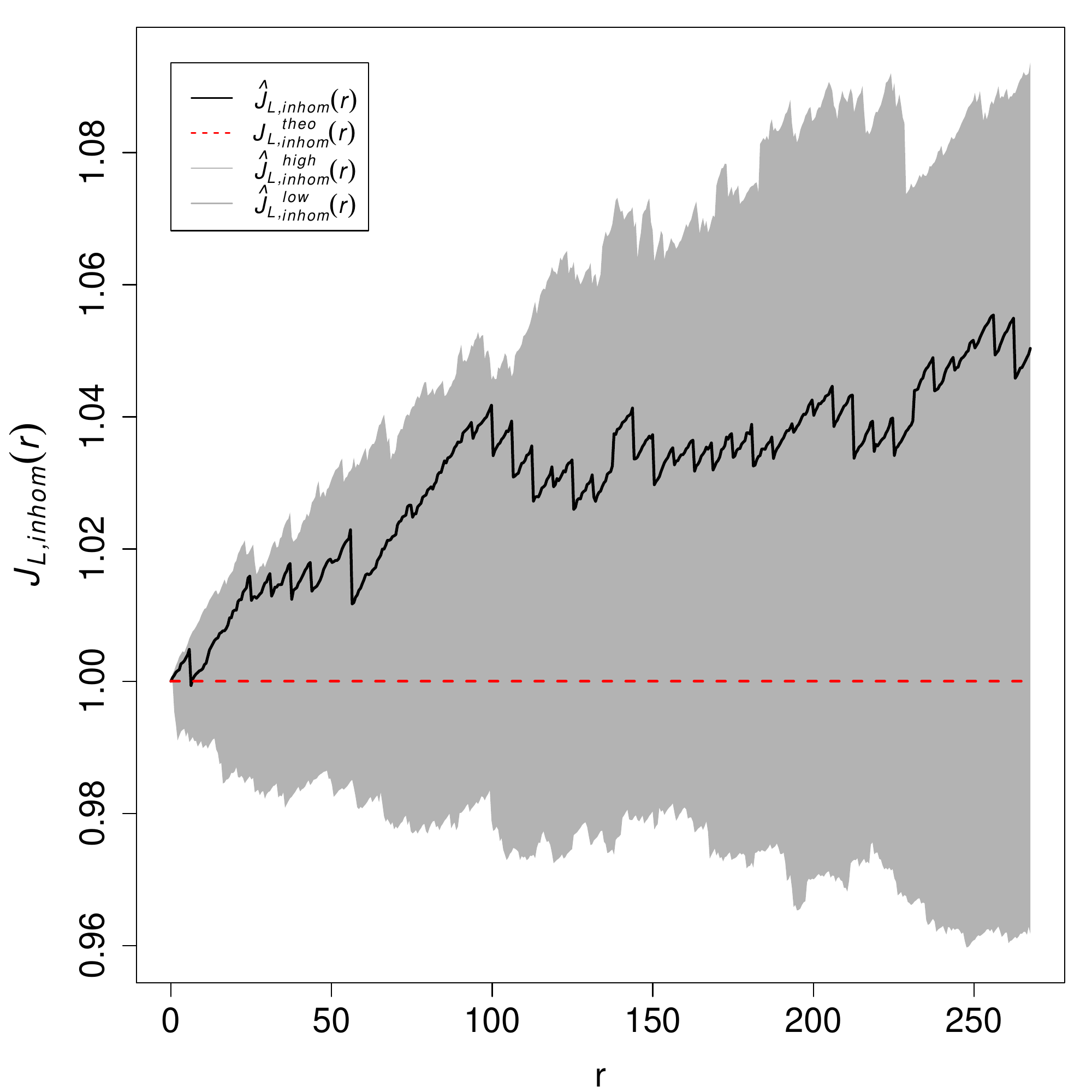}
    \caption{
    % {\em Left}: Motor vehicle traffic accidents in an area of Houston, US, during April, $1999$. {\em Right}: 
    Estimated inhomogeneous linear $J$-functions for a dataset of motor vehicle traffic accidents in an area of Houston, US, during April, $1999$ (left) and a dataset of spiders webs on a brick wall (right). 
    Each plot is displayed together with a pointwise critical envelope based on $99$ simulations of a Poisson process generated from the corresponding estimated intensity. The solid lines are the estimated inhomogeneous linear $J$-functions for the observed patterns and the dashed lines represent the theoretical linear $J$-function for Poisson processes.
    }
    \label{fig:Houston}
\end{figure}

% \begin{figure}[!h]
%     \centering
%     %\includegraphics[scale=.3]{spiders.pdf}
%     \includegraphics[scale=.3]{J-spiders.pdf}
%     \caption{
%     % {\em Left}: Spiders webs on a brick wall. {\em Right}: 
%     The corresponding inhomogeneous $J$-function together with pointwise critical envelopes based on $99$ simulations of a uniform Poisson process. The solid line is the estimated $J$-function for the observed pattern and the dashed line is the $J$-function for Poisson processes.
%     }
%     \label{fig:spiders}
% \end{figure}

\section{Discussion}\label{sec:discuss}
% We have introduced the family of intensity reweighted moment pseudostationary point processes on linear networks.
% \todo[inline]{Discussion}

Methods to statistically analyse point patterns on linear networks/graphs have become increasingly important, as the amount of available linear network point process data has had a steady increase in the last couple of years. Besides univariate analyses, which are carried out by finding intensity estimates for the data, higher-order analyses which detect spatial interaction, i.e.~clustering or inhibition, are central in the (non-parametric) statistical analysis of linear network point processes. 

In Euclidean domains, the most popular tools to carry out analyses of spatial interaction are second-order summary statistics such as inhomogeneous $K$-functions \citep{InhomK2000}. However, when the spatial domain is given by a linear network there immediately arise challenges due the spatially varying local geometry of the network. Early proposals of $K$-functions for linear networks did not take this into consideration, which resulted in erroneous spatial interaction estimates. This issue was finally solved by \citet{ABN12} for the case where the imposed distance/metric on the network was given by the shortest-path distance -- a chosen distance is used to define balls which determine whether two points are within interaction range of one another. These ideas were later extended to a broader class of metrics by \citet{rakshit2017second}, so-called regular distance metrics, with the argument being that the shortest-path distance need not be the canonical distance for a given set of data on a given network. These $K$-functions are referred to as geometrically corrected $K$-functions.

It may be that second-order summary statistics are insufficient to analyse spatial interaction, because the interactions may be more intricate than pairwise interactions. The inhomogeneous nearest neighbour distance distribution function, the empty space function and the $J$-function \citep{van11} have proven themselves to be powerful higher-order summary statistics which may be used to analyse interaction in spatial point processes in Euclidean domains. Hence, one would hope that these could be extended straightforwardly to the linear network context. However, these summary statistics rely on a form of translation invariance of all the (factorial) moments of the underlying point process, which is referred to as intensity reweighted moment stationarity. To make such an extension possible, one would have to define a family of (transitive) transformations on the network in question but this seems unattainable for general networks. We here find a solution to these issues, which consists of i) proposing a new form of (factorial) moment “invariance”, which we refer to as intensity reweighted moment pseudostationarity (IRMPS), and ii) defining geometrically corrected versions of the above higher-order summary statistics, based on regular distance metrics. As a by-product, we obtain a definition of (pseudo)stationarity for linear network point processes as well as geometrically corrected summary statistics for such point processes. With these new summary statistics at hand, we proceed by studying some of their properties and defining non-parametric estimators for them, which we show are (ratio)unbiased when the true intensity function is assumed to be known. We finally evaluate the estimators of our summary statistics numerically, based on simulated data, and use them to analyse two sets of actual linear network point pattern data.

We believe that our new tools may be valuable as alternatives/complements to second-order summary statistics such as $K$-functions. Moreover, our proposed ideas open up for a significant amount of future research. E.g., it would be interesting to characterise which classes of models are IRMPS. In addition, extensions to spatio-temporal and marked point processes on linear networks are also very interesting (cf.~\citet{Cronie2015,Cronie2016}), given the growing amount of available datasets.

\section*{Acknowledgements}
J. Mateu is funded by Grant MTM2016-78917-R from the Spanish Ministry of Economy and Competitivity.

\appendix
\appendixpage

\section{Proofs}\label{s:Proofs}

\begin{proof}[Proof of Theorem \ref{ThmRepresentation}]
We start with $H_{\rm inhom}^L(r;u)$ and note that we only need to show that $K_{\rm inhom}^{L,m}(r;u)$ does not depend on $u\in L$ for an arbitrary $m$ (note also the use of \eqref{eq:intregulardistpairsmult}):% Assuming $k$ORPS, we have
\begin{align*}
& K_{\rm inhom}^{L,m}(r;u)
=
\\
=&
\int_{b_L(u,r)^m}
g_{m+1}(u,u_1,\ldots,u_m)
\prod_{i=1}^m w_{d_L}(u,d_L(u,u_i))
\de_1u_1\cdots\de_1u_m
\\
=&
\int_{b_L(u,r)^m}
\bar g_{m+1}(d_L(u,u),d_L(u,u_1),\ldots,d_L(u,u_m))
\times
\\
&\times
\prod_{i=1}^m w_{d_L}(u,d_L(u,u_i))
\de_1u_1\cdots\de_1u_m
\\
=&
\int_0^{r}
\cdots
\int_0^{r}
\bar g_{m+1}(0,t_1,\ldots,t_m)
\de t_1\cdots\de t_m
.
\end{align*}
Taking the above equality into account
and recalling \eqref{HinhomL} we see that $H_{\rm inhom}^L(r;u)$ does not depend on $u\in L$ for IRMPS point processes. Turning to $F_{\rm inhom}^L(r;u)$, we have that
\begin{align*}
&F^L_{\rm inhom}(r;u) 
= 
\\
=&
1 - \E\left[
\prod_{x \in X} \left(
1 -
\frac{\bar\rho\1\{x\in b_L(u,r)\}w_{d_L}(u,d_L(u,x))}{\rho(x)}
\right)
\right]
\\
=&
-\sum_{m=1}^{\infty}
\frac{(-\bar\rho)^m}{m!}
\int_{b_L(u,r)^m}
g_m(u_1,\ldots,u_m)
\prod_{i=1}^m w_{d_L}(u,d_L(u,u_i))
\de_1u_1\cdots\de_1u_m,
\\
=&
-\sum_{m=1}^{\infty}
\frac{(-\bar{\rho})^m}{m!}
\int_{b_L(u,r)^m}
\bar g_m(d_L(u,u_1),\ldots,d_L(u,u_m)) 
\times
\\
&\times
\prod_{i=1}^m w_{d_L}(u,d_L(u,u_i))
\de_1u_1\cdots\de_1u_m
\\
=&
-\sum_{m=1}^{\infty}
\frac{(-\bar{\rho})^m}{m!}
\int_0^r \cdots \int_0^r 
\bar g_{m}(t_1,\ldots,t_m)
\de t_1 \cdots \de t_m.
\end{align*}
Since both $H^L_{inhom}(r;u)$ and $F^L_{inhom}(r;u)$ do not depend on $u$, we finally conclude that $J^L_{inhom}(r;u)=J^L_{inhom}(r)=
(1-H_{\rm inhom}^L(r))/(1-F_{\rm inhom}^L(r))$,  following the steps in the proof of \citet[Theorem 1]{van11}. 
\end{proof}

\begin{proof}[Proof of Lemma \ref{LemmaThinning}]
The correlation functions $g_m$, $m\geq 1$, are invariant under thinning since $\rho^{(m)}_{th}(u_1,\ldots,u_m)= \rho^{(m)}(u_1,\ldots,u_m)\prod_{i=1}^{m} p(u_i) $ where $\rho^{(m)}_{th}$, $m\geq 1$, are the product densities of $X_{th}$. 

Next, we have that $1-F_{\rm inhom}^L(r)$ and $1-H_{\rm inhom}^L(r)$ coincide with the generating functionals of $X$ and the reduced Palm process $X_u^!$, respectively, when evaluated in the function $h(x) = 1 - \bar\rho\1\{x\in b_L(u,r)\}w_{d_L}(u,d_L(u,x))/\rho(x)$. Exploiting  \citet[equation (11.3.2)]{DVJ2}, we find that $1-F_{\rm inhom}^{L,th}(r)$ and $1-H_{\rm inhom}^{L,th}(r)$ are given by the same generating functionals, but instead evaluated in the function $x\mapsto 1-p(x)+p(x)h(x)$. Hence, they may alternatively be expressed as the indicated expectations of products. 

\end{proof}

\begin{proof}[Proof of Theorem \ref{ThmUnbiased}]
We first start with $\widehat F_{\rm inhom}^L(r)$ and note that
\begin{eqnarray*}
\E\left[\widehat F_{\rm inhom}^L(r)\right]
&=& 
1 - \frac{1}{N(I \cap L_{\ominus r})} 
\sum_{u \in I \cap L_{\ominus r} } 
\E\left[ \prod_{x \in X \cap b_L(u,r)} 
\left( 1- \frac{\bar{\rho}w_{d_L}(u,d_L(u,x))}{\rho(x)}\right)
\right]\\
&=&
1 - \frac{1}{N(I \cap L_{\ominus r})} 
\sum_{u \in I \cap L_{\ominus r} }(1-F_{\rm inhom}^L(r;u))
\\
&=&
\frac{1}{N(I \cap L_{\ominus r})} 
\sum_{u \in I \cap L_{\ominus r} }F_{\rm inhom}^L(r;u)
\end{eqnarray*}
by Definition \ref{SumStatsL}. By Theorem \ref{ThmRepresentation} all of the summands in the last sum are equal (a.e.) and this yields that the whole expression equals $F_{\rm inhom}^L(r)$, since $X$ is IRMPS.

% \[
% F_{\rm inhom}^L(r;u) 
% = 
% 1 - \E\left[
% \prod_{x \in X} \left(
% 1 -
% \frac{\bar\rho\1\{x\in b_L(u,r)\}w_{d_L}(u,d_L(u,x))}{\rho(x)}
% \right)
% \right]
% \]

% where the expectation on the right hand side equals
% \begin{eqnarray*}
% 1+\sum_{m=1}^{\infty}
% \frac{(-\bar\rho)^m}{m!}
% \int_{b_L(u,r)^m}
% g_m(u_1,\ldots,u_m)
% \prod_{i=1}^m w_{d_L}(u,d_L(u,u_i))
% \de_1u_1\cdots\de_1u_m
% \end{eqnarray*}
% by Definition \ref{SumStatsL}. Noting that the last expression does not depend on $u$ since $X$ is IRMPS (recall Theorem \ref{ThmRepresentation}), we obtain that 
% % whereby the expectation o
% % and it then results that 
% $\widehat{F_{\rm inhom}^L}(r)$ is an unbiased estimator. 

% {\color{red}
% THIS IS REALLY UNCLEAR (AND POSSIBLY INCORRECT):

Turning to $\widehat H_{\rm inhom}^L(r)$, by \eqref{RedCM} and Theorem \ref{ThmRepresentation} we have
\begin{align*}
&\E\left[
\sum_{u \in X \cap L_{\ominus r} } 
\prod_{x \in X\setminus\{u \} \cap b_L(u,r)} 
\left( 1- \frac{\bar{\rho}}{\rho(x)} w_{d_L}(u,d_L(u,x))\right)
\right]
=\\
&=
\int_{L_{\ominus r}}
\E_u^!\left[
\prod_{x \in X\setminus\{u \} \cap b_L(u,r)} 
\left( 1- \frac{\bar{\rho}}{\rho(x)} w_{d_L}(u,d_L(u,x))\right)
\right]
\rho(u)\de_1u
\\
&=
\int_{L_{\ominus r}}
H_{\rm inhom}^L(r;u)
\rho(u)\de_1u
=
H_{\rm inhom}^L(r)
\int_{L_{\ominus r}}
\rho(u)\de_1u.
\end{align*}
The expectation of the denominator of \eqref{eq:estHinhom} is $\int_{L_{\ominus r}}
\rho(u)\de_1u$, which yields the ratio-unbiasedness.

% \begin{align*}
% & \E\left[\sum_{u \in X \cap L_{\ominus r} } \prod_{x \in X\setminus\{u \} \cap b_L(u,r)} \left[ 1- \bar{\rho} w_{d_L}(u,d_L(u,x))/\rho(x) \right] 
% \right]
% =
% \\
% &
% \E_u^!\left[
% \prod_{x \in X} \left(
% 1 -
% \frac{\bar\rho\1\{x\in b_L(u,r)\}w_{d_L}(u,d_L(u,x))}{\rho(x)}
% \right)
% \right] \int_{L_{\ominus r}} \rho(u) \de_1 u,
% \end{align*}
% that proves \eqref{eq:estHinhom} is a ratio-unbiased estimator for $H_{\rm inhom}^L(r)$.
% }
\end{proof}

\bibliographystyle{dcu}
\bibliography{main}

@article{moller1998lgcp,
  title={Log {G}aussian {C}ox processes},
  author={M{\o}ller, Jesper and Syversveen, Anne Randi and Waagepetersen, Rasmus Plenge},
  journal={Scandinavian Journal of Statistics},
  volume={25},
  number={3},
  pages={451--482},
  year={1998},
  publisher={Wiley}
}

@article{cronie2018bandwidth,
  title={A non-model-based approach to bandwidth selection for kernel estimators of spatial intensity functions},
  author={Cronie, Ottmar and van Lieshout, Maria Nicolette Margaretha},
  journal={Biometrika},
  volume={105},
  number={2},
  pages={455--462},
  year={2018},
  publisher={Oxford University Press}
}

@article{ABN12,
  title={Geometrically corrected second order analysis of events on a linear network, with applications to ecology and criminology},
  author={Ang, Qi Wei and Baddeley, Adrian and Nair, Gopalan},
  journal={Scandinavian Journal of Statistics},
  volume={39},
  number={4},
  pages={591--617},
  year={2012},
  publisher={Wiley Online Library}
}

@book{BRT15,
	title={Spatial Point Patterns: Methodology and Applications with R},
	author={Baddeley, Adrian and Rubak, Ege and Turner, Rolf},
	year={2015},
	publisher={CRC Press}
}

@article{OSS09,
	title={A kernel density estimation method for networks, its computational method and a GIS-based tool},
	author={Okabe, Atsuyuki and Satoh, Toshiaki and Sugihara, Kokichi},
	journal={International Journal of Geographical Information Science},
	volume={23},
	number={1},
	pages={7--32},
	year={2009},
	publisher={Taylor \& Francis}
}

@book{OS12,
	title={Spatial Analysis along Networks: Statistical and Computational Methods},
	author={Okabe, Atsuyuki and Sugihara, Kokichi},
	year={2012},
	publisher={John Wiley \& Sons}
}

@book{CSKWM13,
	title={Stochastic geometry and its applications},
	author={Chiu, Sung Nok and Stoyan, Dietrich and Kendall, Wilfrid S and Mecke, Joseph},
	year={2013},
	publisher={John Wiley \& Sons}
}

@article{OY01,
  title={The {K}-function method on a network and its computational implementation},
  author={Okabe, Atsuyuki and Yamada, Ikuho},
  journal={Geographical Analysis},
  volume={33},
  number={3},
  pages={271--290},
  year={2001},
  publisher={Wiley Online Library}
}

@book{IPSS08,
  Title = {Statistical Analysis and Modelling of Spatial Point Patterns},
  Author = {Illian, J. and Penttinen, A. and Stoyan, H. and Stoyan, D.},
  Publisher = {John Wiley \& Sons},
  Year = {2008}
}

@book{MW04,
	title={Statistical Inference and Simulation for Spatial Point Processes},
	author={M{\o}ller, Jesper and Waagepetersen, Rasmus},
	year={2004},
	publisher={CRC Press}
}

@article{VB96,
  title={A nonparametric measure of spatial interaction in point patterns},
  author={van Lieshout, M.~N.~M. and Baddeley, A},
  journal={Statistica Neerlandica},
  volume={50},
  number={3},
  pages={344--361},
  year={1996},
  publisher={Wiley Online Library}
}

@article{van11,
  title={A {J}-function for inhomogeneous point processes},
  author={van Lieshout, M.~N.~M.},
  journal={Statistica Neerlandica},
  volume={65},
  number={2},
  pages={183--201},
  year={2011},
  publisher={Wiley Online Library}
}

@article{InhomK2000,
title = "Non-and semi-parametric estimation of interaction in inhomogeneous point patterns",
author = "Adrian Baddeley and J. M{\o}ller and R. Waagepetersen",
year = "2000",
doi = "10.1111/1467-9574.00144",
volume = "54",
pages = "329--350",
journal = "Statistica Neerlandica",
issn = "0039-0402",
publisher = "Wiley-Blackwell",
number = "3",
}

@Article{Cronie2016,
author="Cronie, O. and van Lieshout, M. N. M.",
title="Summary statistics for inhomogeneous marked point processes",
journal="Annals of the Institute of Statistical Mathematics",
year="2016",
volume="68",
pages="905--928"
}

@Article{Cronie2015,
author="Cronie, O. and van Lieshout, M. N. M.",
title={A {J}-function for Inhomogeneous Spatio-temporal Point Processes},
journal="Scandinavian Journal of Statistics",
year="2015",
volume="42",
pages="562--579"
}

@book{DVJ2,
  title={An Introduction to the Theory of Point Processes: Volume II: General Theory and Structure},
  author={Daley, Daryl J and Vere-Jones, David},
  year={2008},
  publisher={Springer-Verlag New York},
  edition={Second}
}

@article{MFJ18,
	title={On kernel-based intensity estimation of spatial point patterns on linear networks},
	author={Moradi, M. Mehdi and Rodriguez-Cortes, Francisco and Mateu, Jorge},
	journal={Journal of Computational and Graphical Statistics},
  volume={27},
  number={2},
  pages={302--311},
  year={2018},
  publisher={Taylor \& Francis}
}

@article{mcswiggan2017,
  title={Kernel density estimation on a linear network},
  author={McSwiggan, Greg and Baddeley, Adrian and Nair, Gopalan},
  journal={Scandinavian Journal of Statistics},
  volume={44},
  number={2},
  pages={324--345},
  year={2017},
  publisher={Wiley Online Library}
}

@article{baddeley14,
  title={Multitype point process analysis of spines on the dendrite network of a neuron},
  author={Baddeley, Adrian and Jammalamadaka, Aruna and Nair, Gopalan},
  journal={Journal of the Royal Statistical Society: Series C (Applied Statistics)},
  volume={63},
  number={5},
  pages={673--694},
  year={2014},
  publisher={Wiley Online Library}
}

@article{diggle79,
  title={On parameter estimation and goodness-of-fit testing for spatial point patterns},
  author={Diggle, Peter},
  journal={Biometrics},
  pages={87--101},
  year={1979},
  publisher={JSTOR}
}

@article{bartlett1964spectral,
  title={The spectral analysis of two-dimensional point processes},
  author={Bartlett, MS},
  journal={Biometrika},
  volume={51},
  number={3/4},
  pages={299--311},
  year={1964},
  publisher={JSTOR}
}

@article{paloheimo1971theory,
  title={On a theory of search},
  author={Paloheimo, JE},
  journal={Biometrika},
  volume={58},
  number={1},
  pages={61--75},
  year={1971},
  publisher={Oxford University Press}
}

@article{baddeley2004spatstat,
   author = {Adrian Baddeley and Rolf Turner},
   title = {spatstat: An \textsf{R} Package for Analyzing Spatial Point Patterns},
   journal = {Journal of Statistical Software},
   volume = {12},
   number = {6},
   year = {2005},
   issn = {1548-7660},
   pages = {1--42}
}

@Article{Moradi2019,
author="Moradi, M. Mehdi
and Cronie, Ottmar
and Rubak, Ege
and Lachieze-Rey, Raphael
and Mateu, Jorge
and Baddeley, Adrian",
title="Resample-smoothing of Voronoi intensity estimators",
journal="Statistics and Computing",
volume = "29",
number = "5",
pages = "995–-1010",
year="2019",
month="Jan",
day="19",
issn="1573-1375"
}

@phdthesis{voss1999habitat,
  title={Habitat Preferences and Spatial Dynamics of the Urban Wall Spider: Oecobius Annulipes Lucas},
  author={Voss, S},
  year={1999},
  school={Honours thesis, Department of Zoology, University of Western Australia}
}

@article{rakshit2019fast,
  title={Fast kernel smoothing of point patterns on a large network using two-dimensional convolution},
  author={Rakshit, S. and Davies, T. M. and Moradi, M. Mehdi and McSwiggan, G and Nair, G and Mateu, J and Baddeley, A},
  journal={International Statistical Review},
  note={doi: 10.1111/insr.12327},
  year={2019}
}

@article{baddeley2017stationary,
  title={“Stationary” point processes are uncommon on linear networks},
  author={Baddeley, Adrian and Nair, Gopalan and Rakshit, Suman and McSwiggan, Greg},
  journal={Stat},
  volume={6},
  number={1},
  pages={68--78},
  year={2017},
  publisher={Wiley Online Library}
}

@article{Zessin,
  title={The method of moments for random measures},
  author={Zessin, Hans},
  journal={Zeitschrift f{\"u}r Wahrscheinlichkeitstheorie und verwandte Gebiete},
  volume={62},
  number={3},
  pages={395--409},
  year={1983},
  publisher={Springer}
}

@article{rakshit2017second,
  title={Second-order analysis of point patterns on a network using any distance metric},
  author={Rakshit, Suman and Nair, Gopalan and Baddeley, Adrian},
  journal={Spatial Statistics},
  volume={22},
  pages={129--154},
  year={2017},
  publisher={Elsevier}
  }

@article{XZY08,
	title={Kernel density estimation of traffic accidents in a network space},
	author={Xie, Zhixiao and Yan, Jun},
	journal={Computers, Environment and Urban Systems},
	volume={32},
	number={5},
	pages={396--406},
	year={2008},
	publisher={Elsevier}
}

@incollection{B05,
	title={Network density estimation: analysis of point patterns over a network},
	author={Borruso, Giuseppe},
	booktitle={Computational Science and Its Applications--ICCSA 2005},
	pages={126--132},
	year={2005},
	publisher={Springer}
}

@article{B08,
	title={Network density estimation: a {GIS} approach for analysing point patterns in a network space},
	author={Borruso, Giuseppe},
	journal={Transactions in {GIS}},
	volume={12},
	number={3},
	pages={377--402},
	year={2008},
	publisher={Wiley Online Library}
}

@article{moradispacetime,
	title={First and second-order characteristics of
spatio-temporal point processes on linear
networks},
	author={Moradi, M. Mehdi and Mateu, Jorge},
	journal={Submitted for publication},
  volume={},
  number={},
  pages={},
  year={2019},
  publisher={Taylor \& Francis}
}

@phdthesis{moradi2018spatial,
    title    = {Spatial and Spatio-Temporal Point Patterns on Linear Networks},
    school   = {University Jaume I},
    author   = {Moradi, M Mehdi},
    year     = {2018},
    type     = {{PhD} dissertation},
}

@book{scott2015multivariate,
  title={Multivariate Density Estimation: Theory, Practice, and Visualization},
  author={Scott, David W},
  year={2015},
  publisher={John Wiley \& Sons}
}

@article{levine2006,
  title={Houston, {T}exas, {M}etropolitan {T}raffic {S}afety {P}lanning {P}rogram},
  author={Levine, Ned},
  journal={Transportation Research Record: Journal of the Transportation Research Board},
  volume={1969},
  pages={92--100},
  year={2006},
  publisher={Transportation Research Board of the National Academies}
}

@incollection{levine2009,
 author={Levine, Ned},
editor={Geertman, Stan and Stillwell, John},
title={A {M}otor {V}ehicle {S}afety {P}lanning {S}upport {S}ystem: The {H}ouston {E}xperience},
bookTitle={{P}lanning {S}upport {S}ystems {B}est {P}ractice and {N}ew {M}ethods},
year={2009},
publisher={Springer Netherlands},
address={Dordrecht},
pages={93--111},
isbn={978-1-4020-8952-7}
}

@article{eckardt2018point,
  title={Point patterns occurring on complex structures in space and space-time: An alternative network approach},
  author={Eckardt, Matthias and Mateu, Jorge},
  journal={Journal of Computational and Graphical Statistics},
  volume={27},
  number={2},
  pages={312--322},
  year={2018},
  publisher={Taylor \& Francis}
}

@article{van2006Mark,
  title={A {J}-function for marked point patterns},
  author={van Lieshout, M. N. M.},
  journal={Annals of the Institute of Statistical Mathematics},
  volume={58},
  number={2},
  pages={235},
  year={2006},
  publisher={Springer}
}

%\newpage
%\todototoc
%\listoftodos
\end{document}